%% file: main.tex
\documentclass[runningheads]{llncs}
\usepackage[T1]{fontenc}
\usepackage{graphicx}

\input{preamble.tex}
\begin{document}
\title{A Formal Semantics of C with OpenMP Parallelism}
%
%




\author{Ke Du\inst{1}\orcidID{0009-0008-2465-1082} \and
Anshu Sharma\inst{2}\orcidID{0000-0002-8686-5835} \and
Liyi Li\inst{3}\orcidID{0000-0001-8184-0244} \and
William Mansky\inst{1}\orcidID{0000-0002-5351-895X}
}
\authorrunning{K. Du et al.}
%
\institute{University of Illinois Chicago\\
 \and
The College of William and Mary
\and
Iowa State University
}
\maketitle              
\begin{abstract}
  OpenMP is a popular parallelization framework that lets users transform sequential code into parallel code with a few simple annotations. Unfortunately, it is also easy to inadvertently introduce errors by adding OpenMP pragmas into otherwise correct programs, including both logic errors and race conditions. We present a formal semantics for C code with OpenMP directives, building on the C semantics of the CompCert verified compiler and its extension to concurrency. Our semantics captures subtle interactions between OpenMP directives and variable state that have been obscured by previous OpenMP semantics, and provides a basis for detecting undesired behaviors introduced by incorrect annotations: in particular, any successful execution is guaranteed to be free of data races.

\end{abstract}
\input{intros.tex}
\input{background.tex}
\input{formal.tex}
\input{comparison.tex}
\input{conclusion.tex}

\newpage\appendix
\section{ClightOMP Syntax}
\input{syntax.tex}

\section{Nested Parallel Region}
\input{nested_par_region.tex}

\section{Team Tree Operations}
\input{ttree_op.tex}

\vspace{-0.2cm}
\section{Supporting Privatization and Reduction Operations}
\input{priv_red_supporting_functions.tex}

\section{An Example ClightOMP Execution}
\input{example_execution.tex}

\section{Data Race Freedom}
\input{drf.tex}

\newpage

\bibliographystyle{splncs04}
\bibliography{sources, paper}

\end{document}

%% file: preamble.tex
\usepackage{filecontents}
\usepackage{amsmath}
\usepackage{amssymb}
\usepackage{ebproof}
\usepackage{todonotes}
\usepackage{algorithmic}
\usepackage{textcomp}
\usepackage{mathpartir}
\usepackage{color}
\usepackage{xcolor}
\usepackage{circuitikz}
\DeclareMathAlphabet{\mathpzc}{OT1}{pzc}{m}{it}

\usepackage{subcaption} 
\usepackage{epsfig}
\usepackage{listings}
\usepackage[override]{cmtt}
\usepackage{graphicx}
\usepackage{balance}
\usepackage{xspace}
\usepackage{booktabs}
\usepackage{url}
\usepackage{xparse}
\usepackage{enumitem}
\usepackage{wasysym}
\usepackage{tikz}
\usetikzlibrary{chains,fit,shapes,positioning,calc,arrows.meta,shapes.arrows}

\usepackage{stackengine,scalerel}
\usepackage{calc}

\usepackage{float}
\PassOptionsToPackage{hyphens}{url}\usepackage{hyperref}
\usepackage{cleveref}
\usepackage{mathtools}

\usepackage{appendix}

\colorlet{kwd}{black!80!green}
\definecolor{spec1}{RGB}{78, 131, 162}
\definecolor{spec0}{RGB}{66, 102, 136}
\definecolor{lespec}{RGB}{30, 80, 180}
\colorlet{spec}{lespec}
\colorlet{auto}{lespec!35!lightgray}

\lstdefinestyle{customc}{
  belowcaptionskip=1\baselineskip,
  breaklines=true,
  breakatwhitespace, 
  numbers=left,
  xleftmargin=\parindent,
  language=C,
  columns=flexible,      
  showstringspaces=false,
  basicstyle=\footnotesize\ttfamily,
  otherkeywords={uint,uintptr_t,let,assert},
  literate={{<-}{{$\leftarrow\,$}}2
            {->}{{$\rightarrow\,$}}2
            {<=}{{$\leq\,$}}2},
  numberstyle=\tiny\rmfamily,
  keywordstyle=\bfseries\color{green!40!black},
  commentstyle=\itshape\color{purple!40!black},
  identifierstyle=\color{blue!80!black},
}
\newenvironment{DIFnomarkup}{}{}
\newcommand{\myparagraph}[1]{\textbf{#1}.\xspace}

\newcommand{\code}[1]{\lstinline[keywordstyle=\color{green!40!black}, basicstyle=\color{green!40!black}, identifierstyle=\color{green!40!black}]|#1|}
\newcommand{\ignore}[1]{{}}

\newif\ifsubmit\submitfalse
\submitfalse
\ifsubmit
\newcommand{\kedu}[1]{}
\newcommand{\liyi}[1]{}
\newcommand{\wm}[1]{}
\newcommand{\anshu}[1]{}
\newcommand{\review}[1]{}

\else
\newcommand{\kedu}[1]{\textcolor{red}{Ke Du: #1}}
\newcommand{\liyi}[1]{\textbf{\textcolor{orange}{Liyi: #1}}}
\newcommand{\wm}[1]{\textbf{\textcolor{teal}{William: #1}}}
\newcommand{\review}[1]{\textbf{\textcolor{blue}{Review: #1}}}
\newcommand{\anshu}[1]{\textbf{\textcolor{red}{Anshu: #1}}}
\fi

\newcommand{\lang}{ClightOMP\xspace}

\newtheorem{defi}{Definition}

\newcommand{\kw}[1]{\ensuremath{\mathtt{#1}}}

\newcommand{\ttree}{\ensuremath{\mathcal{T}}}

\renewcommand{\vec}[1]{\ensuremath{\overrightarrow{\mkern-7mu #1 \mkern-5mu}}}
\newcommand{\vect}{\ensuremath{\vec{t}}}
\newcommand{\IN}{\ensuremath{\in\mkern-6mu}}

\newcommand{\fnt}[1]{\ensuremath{\operatorname{\code{#1}}}}

\newcommand{\cdt}[1]{\texttt{#1}}

\newcommand{\tp}{\ensuremath{\mathcal{P}}}
\newcommand{\tuple}[1]{\langle #1 \rangle}
\newcommand{\Tuple}[1]{\big\langle #1 \big\rangle}
\newcommand{\TUPLE}[1]{\Big\langle #1 \Big\rangle}
\newcommand{\Bracket}[1]{\Big[ #1 \Big]}
\newcommand{\op}{\fnt{op}}
\newcommand{\conc}{\mathbin{||}}
\newcommand{\XMAPSTO}[1]{\xmapsto[#1]{\phantom{a}}}
\newcommand{\lenv}{\mathit{le}}

\newcommand{\upd}{\ensuremath{::=}}

\newcommand{\Team}{\ensuremath{\mathit{Team}}}
\newcommand{\TeamTIDS}{\ensuremath{\fnt{team_tids}}}

\newcommand{\letin}[2]{\mathrm{let\ } #1 \mathtt{\ :=\ } #2 \mathrm{\ in}}

\newcommand{\Nonempty}{\ensuremath{\mathsf{Nonempty}}}
\newcommand{\Readable}{\ensuremath{\mathsf{Readable}}}
\newcommand{\Writable}{\ensuremath{\mathsf{Writable}}}
\newcommand{\Freeable}{\ensuremath{\mathsf{Freeable}}}

\newcommand{\ParCtx}{\ensuremath{\mathsf{ParCtx}}}
\newcommand{\ForCtx}{\ensuremath{\mathsf{ForCtx}}}
\newcommand{\SinCtx}{\ensuremath{\mathsf{SinCtx}}}
\newcommand{\rop}{\mathit{rop}}
\newcommand{\leader}{\ensuremath{\mathit{leader}}}
\newcommand{\ectx}{\ensuremath{\mathit{ectx}}}
\newcommand{\pctx}{\ensuremath{\mathit{pctx}}}

\newcommand{\idx}{\ensuremath{\mathit{idx}}}
\newcommand{\vrc}{\overrightarrow{\mathit{rc}}}
\newcommand{\SSkip}{\cdt{SSkip}}
\newcommand{\SPar}{\cdt{SPar}}

\newcommand{\SFor}{\cdt{SFor}}

\newcommand{\SSin}{\cdt{SSingle}}

\newcommand{\SRed}{\cdt{SRed}}
\newcommand{\SPriv}{\cdt{SPriv}}
\newcommand{\SPrivEnd}{\cdt{SPrivEnd}}
\newcommand{\SBar}{\cdt{SBarrier}}
\newcommand{\SBRB}{\cdt{SBRB}}
\newcommand{\KStop}{\cdt{KStop}}

\newcommand{\IF}{\ensuremath{\mathrm{if}}\>}
\newcommand{\THEN}{\>\ensuremath{\mathrm{then}}\>}
\newcommand{\ELSE}{\>\ensuremath{\mathrm{else}}\>}
\newcommand{\WHEN}{\>\ensuremath{\mathrm{when}}\>}
\newcommand{\WHERE}{\>\ensuremath{\mathrm{where}}\>}
\newcommand{\TT}{\ensuremath{\mathsf{true}}}
\newcommand{\FF}{\ensuremath{\mathsf{false}}}
\newcommand{\id}{\ensuremath{\fnt{id}}}

\newcommand{\StepClight}{\textsc{Step-Clight}}
\newcommand{\StepPriv}{\textsc{Step-Priv}}
\newcommand{\StepPrivEnd}{\textsc{Step-Priv-End}}
\newcommand{\StepRedL}{\textsc{Step-Red-Leader}}
\newcommand{\StepRedNL}{\textsc{Step-Red-Not-Leader}}
\newcommand{\StepPar}{\textsc{Step-Parallel}}
\newcommand{\StepBar}{\textsc{Step-Barrier}}
\newcommand{\Alloc}{\ensuremath{\mathsf{Alloc}}}
\newcommand{\Free}{\ensuremath{\mathsf{Free}}}
\newcommand{\Read}{\ensuremath{\mathsf{Read}}}
\newcommand{\Write}{\ensuremath{\mathsf{Write}}}
\newcommand{\MemEv}{\ensuremath{\mathsf{MemEv}}}
\newcommand{\SyncEv}{\ensuremath{\mathsf{SyncEv}}}
\newcommand{\BarEv}{\ensuremath{\mathsf{Bar}}}
\newcommand{\Par}{\ensuremath{\mathsf{Par}}}
\newcommand{\Event}{\ensuremath{\mathsf{Event}}}

\newlength\shlength
\newcommand\longvec[2][0]{\setlength\shlength{#1pt}%
      \stackengine{-4.6pt}{$#2$}{\smash{$\kern\shlength%
        \stackengine{5.55pt}{$\mathchar"017E$}%
          {\rule{\widthof{$#2$}}{.57pt}\kern.4pt}{O}{r}{F}{F}{L}\kern-\shlength$}}%
          {O}{c}{F}{T}{S}}
      

%% file: intros.tex
\section{Introduction}\label{sec:intros}

The OpenMP framework~\cite{openmp_spec} provides a straightforward entry point for shared-memory parallel programming. Instead of explicitly creating and managing threads and locks, a programmer can write a sequential program and then use the OpenMP API to instruct the compiler to automatically parallelize the code, e.g., by distributing the iterations of a loop across multiple threads. The core of the API is a set of compiler directives (\code{#pragma}s in C and C++) that can be added before the code segments to be parallelized, supported by common compilers including GCC and Clang. However, the simplicity of this interface can be deceptive---OpenMP still suffers from all the complexity of shared-memory concurrency, and choosing the wrong directives can easily introduce race conditions or logic errors.

Prior work~\cite{atzeni} defined formal semantics for OpenMP for the purpose of race detection, modeling the synchronization effects of key directives and the memory accesses they produce. This is useful for detecting concurrency bugs, but does not fully address the question ``what is this OpenMP program allowed to do?'' In this work, we aim to answer that question with a formal semantics for OpenMP at the language level, specifically for C programs with OpenMP annotations. Our semantics builds on CompCert's C semantics~\cite{compcert} and its extension to concurrency~\cite{cpm}, giving us high confidence in our base language semantics; our task is then to accurately capture the allowed behavior of OpenMP directives.

There are three main challenges to giving formal semantics for OpenMP at the language level:

\begin{enumerate}
\item The OpenMP API is not a library of functions but a set of compiler directives that modify the behavior of the code they apply to. These directives cannot be modeled as special function calls executed in a single step, as is standard in concurrent extensions of CompCert, since, at a minimum, they may have effects at both the beginning and the end of the affected block.
\item Unlike standard fork-join concurrency, in which an unstructured collection of threads interacts peer-to-peer via shared memory or message passing, OpenMP describes the behavior of dynamically created \emph{teams} of threads, and the validity and effects of directives depend on the team to which the executing thread currently belongs. Any formal model of OpenMP must track the team structure of the executing threads in addition to more usual forms of state.
\item OpenMP directives have subtle effects on local variables: depending on the position of its declaration and the details of the directive, a variable may be treated as shared memory across multiple threads or split into private copies that may be reconciled at the end of a parallel section---and mistakes in variable-handling directives are a major source of OpenMP-induced bugs.
\end{enumerate}
By carefully modeling the scope of OpenMP directives and their effects on team structure and local variables, in addition to their concurrency and memory behavior, our semantics provide a basis for determining the allowed behavior of an OpenMP-annotated program, and thus ultimately for determining whether a given parallelization preserves the intended behavior of the code. Our semantics is stuck in the presence of memory errors or data races, making it useful for bug-finding. The semantics can also serve as a basis for analysis tools that prove the safety and functional correctness of OpenMP C programs.

\subsection{A Motivating Example}\label{sec:motivationexample}

\begin{figure}[htb]
{\small
\begin{tabular}{l@{\qquad\qquad}l}
{\captionsetup[lstlisting]{margin = 4 mm}
\begin{lstlisting}[xleftmargin= 4 mm, basicstyle=\footnotesize]
// program starts with one thread t1
#pragma omp parallel num_threads(2)
{// a team of 2 threads {t1, t2} is created
  f();
#pragma omp for
  for (i=0; i<6; i+=1) {
    // each thread only works a loop partition
    loop_body
  }
}// end of parallel region, only thread is t1
\end{lstlisting}
}
&
{\captionsetup[lstlisting]{margin = 4 mm}
  \begin{lstlisting}[xleftmargin= 4 mm, basicstyle=\footnotesize]
int k = 0, s = 0;
#pragma omp parallel for
    reduction(+ : k, s)
for (int i = 0; i < 128; i++) {
  k = k + 3;
  s = s + k;
}
\end{lstlisting}
}
\end{tabular}
}
\caption{Left: a typical OpenMP program. Right: an incorrect OpenMP program.}
\label{fig:program-ex}
\end{figure}

We illustrate the uses and challenges of our semantics with a typical OpenMP program (shown on the left of \Cref{fig:program-ex}) and a typical OpenMP mistake (shown on the right).
The core mechanism of OpenMP is a set of \emph{directives}, compiler annotations to parallelize a sequential program, implemented as \code{#pragma}s in C. Some directives are associated with the following code block, binding to them to form a \emph{construct}. Other directives, such as \code{barrier}\!, do not associate with a statement and form a construct by themselves.

In the left program, \code{parallel num\_threads(2)} is an OpenMP directive that specifies that a \emph{team} of two threads should execute the associated code in parallel, with all threads executing the same code block. 
The \code{omp for} directive on line 5 is always associated with a for-loop, and divides the workload of the associated for-loop among the threads of the executing team. When this program executes, it begins with a single thread of execution; at line 2, the \code{parallel} directive creates a team with two threads; and the iterations of the for-loop on line 6 are partitioned across the threads, e.g., each thread may execute three steps of the loop. 

This illustrates several challenges in defining the semantics of OpenMP:

\begin{enumerate}
\item The \code{for} directive is team-sensitive, i.e., the partitioning of the loop at line 5 depends on the number of threads set up in line 2.
\item OpenMP does not check for race-freedom: if \code{loop\_body} is not safe to run in parallel, the program will not execute correctly.
\item At the end of the \code{parallel} block at line 10, there is an implicit barrier: every thread in the team must wait until all threads finish their tasks.
\end{enumerate}
Our semantics must track the team structure throughout the program, detect race conditions among threads, and capture the effects of both the start and end of a construct.

Another difficulty arises from the complex interaction between OpenMP directives and local variables, as shown in the right program in \Cref{fig:program-ex} (adapted from Ahmed et al.~\cite{openmp-diff-test}). When run sequentially, this program sums the multiples of 3 from 3 to 3 * 128. The \code{reduction} clause aims to divide the work across threads; each thread receives its own copies of \code{k} and \code{s}, and the copies are summed at the end of the loop. For \code{k}, this will indeed produce the right result (3 * 128 = 384). However, depending on the distribution of loop iterations among threads, some per-thread values of \code{k} will be reused in computing \code{s}, while other values will never be reached, producing an incorrect final value of \code{s}. This program does not contain a data race---it is not buggy in and of itself---but it does not have the same semantics as the corresponding sequential program, and we can only determine this by modeling the values of \code{k} and \code{s} in each thread. Notably, the main prior work on the semantics of OpenMP~\cite{atzeni} works at the level of memory operations, obscuring the effects of OpenMP directives on variables and control flow. In contrast, our OpenMP semantics directly models the effects of directives on the C program state, allowing us to detect errors like that on the right of \Cref{fig:program-ex}.

\subsection{Contributions and Roadmap}
The main contribution of this paper is, to our knowledge, the first formalization of OpenMP semantics in C at the language level. We formalize 5 constructs and 3 clauses representing commonly used OpenMP directives. In particular, we give the first formal semantics for the \code{for} and \code{single} clauses, which control work-sharing within teams, and the \code{private} and \code{reduction} clauses, which control the interaction between threads and local variables. 
Our novel \emph{team tree} construction captures the essence of OpenMP's execution model and could be reused to define OpenMP semantics in other languages, such as C++ or FORTRAN. Our semantics is based on per-thread permissions to memory, and we prove that any execution in our semantics is data-race-free.

The paper is structured as follows. In \Cref{sec:overview,sec:teamtree}, we describe the syntax of \lang and the team tree model that plays a key role in \lang program states. In \Cref{sec:rules}, we formally present our operational semantic rules for OpenMP and describe in detail the semantics of the major constructs and clauses of OpenMP: parallel and worksharing constructs, barriers, and privatization and reduction clauses. The ClightOMP semantics is formalized in the Rocq proof assistant\footnote{Available in \url{https://github.com/dkxb/ClightOMP}.}. 
\Cref{sec:property} describes the properties of the semantics, including a proof of data-race freedom for any execution. Finally, in section \ref{sec:comparison} we compare with related work, and in section \ref{sec:conclusion} we summarize our results and discuss future work.

%% file: background.tex
\section{Background: The Concurrent Permission Machine}
\label{sec:background}

Since our goal is to give precise semantics for C programs with OpenMP directives, we begin with the authoritative formal semantics of ordinary C, the Clight semantics of the CompCert verified compiler~\cite{compcert}. However, CompCert's model of C does not include concurrency---it is strictly a sequential semantics. Several more recent works, including CASCompCert~\cite{cascompcert} and the Concurrent Permission Machine (CPM)~\cite{cpm}, lift Clight to concurrency by defining a multithreaded machine where each step is either a sequential Clight step or a concurrency operation; our semantics builds on the CPM approach. In this section, we explain the parts of the CPM's program states and step rules that will be necessary to understand our semantics. 

A CPM program state has the form ${\langle \mho, \tp, m\rangle}$, consisting of a schedule $\mho$, a thread pool $\tp$, and a shared memory $m$. The schedule $\mho$ is a list of thread IDs that determines the order in which threads interleave; the possible behaviors of a program are its behaviors under all schedules. The thread pool $\tp$ maps each thread ID $i$ to its local state $P(i)=\tuple{\sigma_i, \pi_i}$, where $\sigma_i=\tuple{s, \lenv, k}$ is a Clight state containing the statement $s$ to be executed, the local variable environment $le$ that maps variable names to their addresses in the memory $m$, and the continuation $k$; $\pi_i$ is the \emph{permission map} of thread $i$, describing the thread's access level to each memory location ($\mathsf{None}$, \Nonempty, \Readable, \Writable, or \Freeable). These access levels allow the CPM to detect data races: at each step of execution, the permissions held by all threads must be compatible with each other, so that e.g. if one thread holds {\Writable} access to a location, no other thread can hold {\Readable} access. Finally, $m$ is a CompCert memory \cite{besson2015concrete}. For this paper, it can be viewed as a pair $\tuple{m_v, \pi_m}$ where $m_v$ maps addresses to values, and $\pi_m$ is a placeholder for the executing thread's permission map. We define a few shorthands: $m(l)\triangleq m.1(l)$ for loading a value from address $l$, $\fnt{perm}(m)\triangleq m.2$ for getting the permission map, $m|_{\pi}\triangleq \tuple{m.1, \pi}$ for replacing the permission in $m$ with $\pi$.

The operational semantics of the CPM includes two kinds of steps. First, an individual thread may take a step according to the sequential Clight semantics $\rightarrow_{\textit{sq}}$:
\[
\inferrule[Step-Clight]
    {\tp(i)=\tuple{\sigma_i, \pi_i} \and
    \tuple{\sigma_i, m|_{\pi_i}} \rightarrow_{\textit{sq}} \tuple{\sigma'_i, m'} \and
    \tp[i \upd \tuple{\sigma_i', \fnt{perm}(m')}] = \tp'
    }
    {\langle i \cdot \mho, \tp, m\rangle \rightarrow
     \langle \mho, \tp', m' \rangle}
\]
If $i$ is the next thread scheduled to execute, and its current state is $\tuple{\sigma_i, \pi_i}$ where $\sigma_i$ can execute a sequential C statement, then it executes as follows. First, it combines its permissions $\pi_i$ with the shared memory $m$ to obtain a thread-local ``view'' of memory $m|_{\pi_i}$, which has all the same values as $m$ but allows only operations that the thread has permission for (e.g., if $\pi_i$ only has \kw{Readable} access to a location $\ell$, then a write to $\ell$ in $m|_{\pi_i}$ will fail). The local state $\sigma_i$ and local memory $m|_{\pi_i}$ combine to make an ordinary Clight state, which may step according to Clight's sequential semantics to a new state $\tuple{\sigma'_i, m'}$. The new memory $m'$ then becomes the new shared memory, and the thread pool records the new local state $\sigma_i'$ and permissions $\fnt{perm}(m')$ for thread $i$.

The \textsc{Step-Clight} rule applies when the next operation in the scheduled thread is a sequential statement. When it is a concurrency operation, we instead look for a specific rule for that operation in the CPM. For instance, the rule for thread creation is:
\[
\inferrule[Step-Spawn]
    {
        j \not\in \tp \\
        \pi_i = \pi_i' \oplus \pi_j' \\
    }
    {
        \TUPLE{
            i \cdot \mho, \tp \Bracket{
                i\upd \Tuple{\tuple{\cdt{spawn}(f, a), \lenv, k}, \pi_i}
        }, m}
        \rightarrow \\
        \TUPLE{
            \mho, \tp \Bracket{
                i \upd \Tuple{\tuple{\SSkip, \lenv, k}, \pi_i'}, 
                j \upd \Tuple{\tuple{f(a), \emptyset, \KStop}, \pi_j'}
        },  m}
    }
\]
The in-place update notation $\tp[i \upd (\sigma, \pi)]$ is the thread pool $\tp$ but thread $i$ has state $(\sigma, \pi)$. If thread $i$ wants to spawn a new thread executing function $f$ with argument $a$, it can do so by creating a new thread whose code is $f(a)$ with a fresh index $j$ not yet in $\tp$. The parent thread's statement is replaced by the no-op statement $\SSkip$, so its next step will be decided by its continuation $k$. The new thread has an empty local variable environment $\emptyset$ and a continuation $\KStop$, meaning that the thread halts after $f(a)$. The permissions $\pi_i$ of $i$ must be split into two parts, one retained by the parent thread ($\pi_i'$) and one given to the child thread ($\pi_j'$).

The CPM includes rules for thread creation and lock operations (allocation, deallocation, acquire, and release). Taken together, these implement a concurrent extension of Clight with an SC-DRF (Sequential Consistency for Data-Race-Free programs) memory model: programs with data races are stuck (due to permission conflicts), while programs without data races have sequentially consistent behavior, interleaving the steps of individual threads. We develop our OpenMP semantics by extending the CPM with rules for OpenMP directives like parallel regions and work-sharing constructs, extending the state model as necessary. 
OpenMP's memory model is also SC-DRF as long as programs do not use weak-memory \kw{atomic} constructs (which we do not support), so the CPM provides a reasonable basis for capturing the concurrency behavior of C programs with OpenMP directives.


%% file: formal.tex
\definecolor{color1}{HTML}{FEEF5B}
\definecolor{color2}{HTML}{FFA570}
\definecolor{color3}{HTML}{afca54}
\definecolor{blackish}{HTML}{353132}
\definecolor{borderColor}{rgb}{1,0,0}
\tikzset{myArrow/.style={->, >=Straight Barb, line width=1pt, draw=black}}
\tikzset{square/.style={regular polygon,regular polygon sides=4,rounded corners, minimum width=2.5cm, minimum height=2.5cm}}
\tikzset{sq1/.style={square, fill=color1}}
\tikzset{sq2/.style={square, fill=color2}}
\tikzset{sq3/.style={square, fill=color3}}
\tikzstyle{every node}=[font=\Huge]

\section{Syntax and Semantics of ClightOMP}

ClightOMP extends the syntax of Clight~\cite{compcert}, the subset of C formalized in CompCert, with OpenMP directives as new statements. The full syntax is defined in \Cref{fig:ClightO-syn}. In this section, we introduce the syntax and give an operational semantics to these statements, building on Clight and the CPM semantics. 

\subsection{Syntax: Extending Clight}
\label{sec:overview}

The extended statements for constructs, such as the one for the parallel construct, $\SPar_{\mathit{idx}}\ \mathit{nc}\ \mathit{pc}\ \vec{rcs}\ s$, have 4 parts in general: $\SPar$ specifies the kind of directive; $\mathit{nc}$, $\mathit{pc}$ and $\vec{rcs}$ are construct-specific \textit{clauses} that specify additional behaviors; $s$ is the statement affected by the directive; lastly, the subscript $\mathit{idx}$ is an index that uniquely identifies the construct. Concretely, line 2-10 in \Cref{fig:program-ex} is encoded as $\SPar_{\textit{idx}}\ 2\ []\ []\ s$, where $\textit{idx}$ is some unique index and $s$ is the statement for line 3-10.
We explain them in detail.

In addition to \kw{SPar} for the \code{parallel} directive, we define 4 core statements for directives in total: \kw{SFor} for the \code{for} directive, \kw{SSingle} for the \code{single} directive, and \kw{SBarrier} for the \code{barrier} directive.

The \!\code{parallel}\!, \code{for}\!, and \code{single} directives are lexically associated with the following statement (or block) and semantically affect that statement, and we encode the affected statement as the last argument $s$ of \kw{SPar}, \kw{SFor} and \kw{SSingle}. These directives and their associated statements form a \emph{construct}. The \code{barrier} directive does not associate with a statement and forms a barrier construct on its own.

A construct statement can have 0 or more clauses. The $\mathit{nc}$ clause specifies the number of threads in the \kw{parallel} construct\footnote{In the OpenMP standard, the actual number of threads is determined by a number of factors at runtime; currently, we assume that the number is determined statically, but our semantics could be modified to follow the standard more closely.}. A privatization clause $\mathit{pc}$ gives a list of variable names to privatize. Constructs that support reduction have a list of reduction clauses $\ \vec{rcs}$, and a reduction clause $\mathit{rc}$ is a pair of a reduction identifier $\mathit{rid}$ that specifies the operation for combining the reduction contributions, and a list of variables to be reduced---for example, at the end of a for loop that computes a sum in variable \code{sum}, we may sum every thread's contribution to compute the final result using a reduction clause \code{(+:sum)}. Both \kw{SPar} and \kw{SFor} support a privatization clause and a list of reduction clauses; \kw{SSingle} supports privatization but not reduction. When directives with these clauses are executed, we generate runtime-only instructions \kw{SPriv}, \kw{SPrivEnd}, \kw{SRed} that carry a mix of statically and dynamically recorded information about how to perform privatization and reduction; we discuss these further in \Cref{sec:variables}.

There are several places in the semantics of OpenMP where we need to know that statements executed by different threads are not merely equal but in fact the same syntactic occurrence of the statement: for instance, threads in a team synchronize at a barrier only if they have all reached the same barrier. For this purpose, we give every construct statement (\kw{SPar}, \kw{SFor}, \kw{SSingle}, \kw{SBarrier}) a unique integer index $n$. We also generate \kw{SBarrier}s at runtime to make implicit barriers explicit, and they need to be further distinguished with an optional statement $s$ and a boolean $b$; we discuss these details in \Cref{sec:barrier}.

\subsection{Synchronization with the Team Tree}
\label{sec:teamtree}

\ignore{
\begin{figure}[t]
{\small
{\captionsetup[lstlisting]{margin = 4 mm}
  \begin{lstlisting}[xleftmargin= 4 mm]
    // program starts with one thread t1
    #pragma omp parallel num_threads(2)
    {
    %*\colorbox{blueannoback}{@Entity}*)
    // a team of 2 threads {t1, t2} is created
        f();
    #pragma omp for
        for (i=0; i<6; i+=1) {
            // each thread only works a partition of the loop
            loop_body
        }
    }
    // end of parallel region, only thread is t1
  \end{lstlisting}
}
}
\caption{An OpenMP program.}
\label{fig:program-ex}
\end{figure}
}

As described in \Cref{sec:background}, the CPM models the execution of concurrent programs by tracking a collection of per-thread states and a single shared memory. For OpenMP, we need to track an additional kind of information: the hierarchical structure of spawned threads, which are organized into \emph{teams}, and the states of their in-progress OpenMP operations. Since the team hierarchy in the program forms a tree (each team is spawned by a single parent thread, which may itself be part of a team), we organize OpenMP-specific state information into a tree that we call the \emph{team tree}, denoted by \ttree, which we include in program states alongside the thread pool $\tp$ and memory $m$. The execution of sequential C code (rule \textsc{Step-Clight}) is not affected by the addition of the team tree, as regular per-thread execution is independent of OpenMP states.


We define the team tree by defining its $\mathit{Node}$s. Each node corresponds to a thread of execution.
When a thread spawns a team, its node adds child nodes for the team's threads, along with some \emph{team contexts} that record OpenMP state common to the entire team. Formally:

\begin{defi}[Node and Team]
\emph{
A $\mathit{Node}$ of a team tree is a tuple $(t, \mathit{tm})$ where $t$ is a thread ID
and $\mathit{tm}$ is either some $\Team$ or $\bot$. A $\Team$ is a tuple $(\leader, \pctx,\ectx, \overrightarrow{\mathit{mates}})$ where $\leader$ stores the thread ID for the leader of the team (which is always the thread that spawns the team), $\pctx$ is the parallel context of the team, $\ectx$ is a team-executed context representing the currently active team-executed construct\footnote{One team-executed context suffices because all threads in a team can be in at most one team-executed construct at a time: the OpenMP standard effectively requires that, inside a parallel region, team-executed contexts cannot be nested
. From the standard: ``a team-executed region (\kw{for}, \kw{single}, \kw{barrier}) may not be closely nested (nested in the same parallel region) inside a partitioned worksharing region (\kw{for} or \kw{single})''. We implement the \kw{barrier} construct as a single atomic operation, and all of our team-executed constructs have an implicit barrier at the end of the construct, so no thread can enter a new team-executed construct while some of its teammates are still executing the old one. } or $\bot$, and $\overrightarrow{\mathit{mates}}$ is a list of $\mathit{Node}$s for the members of the team.
}
\end{defi}

The program starts with the initial team tree $(t_1, \bot)$, containing a single $\mathit{Node}$ for the initial thread $t_1$ and no team.

The parallel context $\ParCtx_{\textit{idx}}$ of a $\Team$ records the index $\textit{idx}$ of the parallel construct that spawned it; $\ectx$ exists if some member of the team is inside a team-executed construct, and records information specific to that construct. It has the form $\ForCtx_\idx\> w$ for a \kw{for} construct, or $\SinCtx_\idx\> j$ for a \kw{single} construct. We will explain the definition and the usage of the contexts when we introduce the semantics in the following sections. 

The thread that spawns a team becomes the \textit{leader} of this team; it participates in this team as a new node with the same thread ID (therefore this new node is also a child of the original node). Spawned threads can also spawn teams themselves, resulting in nested teams.  \footnote{In this work, we only model programs that are executed by one \emph{device}, and all synchronization is due to the OpenMP pragmas that we define. We do not model the \kw{teams} construct, or the notion of a \kw{league} of teams, and all threads are descendants of the initial thread. } Only the leaf node of $i$ in a team tree $\ttree$ is immediately relevant to its execution, and we call this node the \textit{active} node of $i$, denoted as $\ttree(i)$. $\ttree(i)$ is only defined if the active node for $i$ is unique.

In this section, we will describe the operations on team trees informally; formal definitions are in \Cref{fig:team-ops}.

\ignore{
    More specifically, the team executes the \emph{parallel region} of this construct:
    
    \begin{defi}[Informal, OpenMP\footnote{Definitions labeled \textsc{OpenMP} are adapted from the OpenMP standard, rather than being specific to our semantics.}]
    \emph{A \emph{region} is the computations encountered during a specific execution of some construct, as well as the computational content of the start and end of that construct.}
    \end{defi}
    
    \begin{remark}[OpenMP]
        \emph{The \kw{parallel} construct gives rise to a \emph{parallel region}, the \kw{for} construct a \emph{worksharing region}, and the \kw{barrier} construct a \emph{barrier region}. These regions can be nested, but as we mentioned before, a region induced by a team-executed construct cannot be nested immediately within another without a parallel construct in between them. The outermost part of the program is an \emph{implicit} parallel region marked as parallel region I in \Cref{fig:program-nested-par-region}}. 
    \end{remark}
}

\ignore{
    In general, this means that a thread has one node in the tree for each parallel region that it has entered, and only the node for its current parallel region and the team contexts for this region are immediately relevant to its execution. We call such nodes \emph{active}, and other nodes \emph{inactive}:
    
    \begin{defi}
        \emph{A node in the team tree is \emph{active} if it is a leaf node in the team tree.}
    \end{defi}
}

\subsection{Overview of the Semantics}
\label{sec:rules}
Our semantic rules for the OpenMP constructs are shown in \Cref{fig:semantics-2,fig:semantics-3}. As described in the previous section, our step relation is on configurations of the form ${\langle \mho, \tp, \ttree, m\rangle}$, where $\mho$ is a thread schedule, $\tp$ is a collection of per-thread states including control state (for example, the current Clight statement and continuation), local variables, and permissions, $\ttree$ is a team tree, and $m$ is the shared memory.
Each rule begins with a schedule of the form $i \cdot \mho$, indicating that $i$ is the next thread to execute, and then explains the effect on the configuration if $i$ is executing one of the OpenMP directives (regular program steps do not affect the team tree, and so use the \textsc{Step-Clight} rule from \Cref{sec:background} as is). In the rest of this section, we describe each of these rules in detail.

\paragraph{Judgement and Notation}
Although OpenMP constructs are essentially concurrent, some can be described as local transitions of a single thread, together with synchronizations via the team tree $\ttree$. This is expressed with the judgement $s \conc \lenv \conc m \xmapsto[i,\tp, \ttree]{h} s' \conc \lenv' \conc m'$: it says a thread $i$ updates its statement $s$ and variable environment $\lenv$ (along with the global memory $m$) to $s', \lenv'$ (and $m'$), while updating the team tree with $h$, where $h$ may bind the original thread pool $\tp$ and team tree $\ttree$, and the new team tree is $h(\ttree, i)$. We omit $h$ if the team tree is not updated. The $\textsc{Step-Thread}$ rule in \Cref{fig:semantics-2} lifts a single-thread step by thread $i$ to the full configuration, writing back the updated $s'$ and $\lenv'$ to the thread pool $\tp$, updates $\ttree$ to $h(\ttree, i)$, and $m$ to $m'$. Thread $i$'s continuation $k$ stays the same. Like $\textsc{Step-Clight}$, thread $i$ can only step with $m|_{\pi_i}$, i.e. $m$ equipped with thread $i$'s initial permission $\pi_i$; we then save $\fnt{perm}(m')$, the permission in the new memory $m'$, as thread $i$'s new permission. \textsc{Step-Parallel} and \textsc{Step-Barrier} change the state of multiple threads and are not expressed in this style. We abuse the $\lenv|_{\vec{x}}$ notation to mean the map $\lenv$ with the domain restricted to $\vec{x}$.

We explain the basic semantics of the constructs that enable parallelism and worksharing in \Cref{sec:parallelism-and-work-sharing}. Privatization and reduction clauses can be attached to these constructs to fine-tune their behavior; their semantics are explained in \Cref{sec:variables}.

\begin{DIFnomarkup}
\begin{figure}[t]
{
\scriptsize 
\centering
\begin{mathpar}

{\small
    \begin{aligned}
    \SBRB_{\idx,s,r} \triangleq & \mathtt{\ let\ } s_r \mathtt{\ := \ } \IF r.1 = [] \THEN \SSkip \ELSE \SBar_{\idx,\FF,s}; \SRed\>r \mathtt{\ in} \\[-1em]
      & \quad s_{r};\SBar_{\idx,\TT,s}
    \end{aligned}
}

\vspace{1em}

\inferrule[Step-Parallel]
    {
        \forall j \IN \vect. j \notin \tp \\
        |\vect|=nc-1 \\
        r \triangleq (\vrc, \lenv|_{\vec{x}}) \\
        s' \triangleq \SPriv\>{\vec{x}}\> (s;\SBRB_{\idx,\SPar,r}) \\
        \pi_i = \pi_i' \oplus (\bigoplus\nolimits_{j\IN \vect} \pi_j') \\ 
    }
    {
      \TUPLE{i \cdot \mho, \tp\Bracket{i\upd \Tuple{\tuple{\cdt{SPar}_\mathit{idx}\ \textit{nc}\ \vec{x}\ \vrc\ s, \lenv, k}, \pi_i}}, \ttree, m}
      \rightarrow \\
      \TUPLE{ \mho, \tp\Bracket{i\upd \Tuple{\tuple{s', \lenv, k}, \pi_i'}, \forall j \IN \vect.\> j \upd \Tuple{\tuple{s', \lenv, \KStop}, \pi_j'}},  \fnt{spawn\_team}(\ParCtx_{\mathit{idx}}, \vect)(\ttree, i), m}
    }

\inferrule[Step-Barrier]
    {
        \vect \triangleq \TeamTIDS (\ttree, i) \\
        \ttree(i).\fnt{ectx} \neq \bot \iff s=\SFor \lor s=\SSin \\
        \bigoplus\nolimits_{j\IN \vect} \pi_j = \bigoplus\nolimits_{j\IN \vect} \pi'_j
    }
    {
        \TUPLE{
            i \cdot \mho, \tp\Bracket{\forall j \IN \vect, j\upd \Tuple{\tuple{\cdt{SBarrier}_\mathit{idx, b, s}, \lenv_j, k_j}, \pi_j}},  \ttree, m
        }   
        \rightarrow \\
        \TUPLE{
            \mho,\tp\Bracket{\forall j \IN \vect, j\upd \Tuple{\tuple{\SSkip, \lenv_j, k_j}, \pi_j'}}, \IF b \THEN \fnt{end_ctx}(s)(\ttree, i) \ELSE \ttree, m\rangle
        }
    }


\inferrule[Step-Thread]
    {
        s \conc \lenv \conc m|_{\pi_i} \xmapsto[i,\tp, \ttree]{\mathit{h}} s' \conc \lenv' \conc m'  \\
    }
    {
      \langle i \cdot \mho, \tp\Bracket{i \upd \Tuple{\tuple{s, \lenv, k}, \pi_i}}, \ttree, m\rangle
      \rightarrow \langle
      \mho, \tp\Bracket{i \upd \Tuple{\tuple{s', \lenv', k}, \fnt{perm}(m')}}, \mathit{h}(\ttree, i), m'\rangle
    }

\inferrule[Step-For]
    {
        w \triangleq \fnt{partition}(s) \\
        s' \triangleq \fnt{set\_stmt\_to\_partition}(s, w(i)) \\
        r \triangleq \tuple{\vrc, \lenv} \\
        s'' \triangleq \SPriv\>{\vec{x}}\> (s'; \SBRB_{\idx,\SFor, r}) \\
    }
    {
        \cdt{SFor}_\mathit{idx}\> \vec{x}\> \vrc\> s \conc \lenv \conc m  \xmapsto[i]{\fnt{sync\_ectx}(\ForCtx_\mathit{idx}\> w)} s' \conc \lenv \conc m
    }

\inferrule[Step-Single]
    {
        s' \triangleq \SPriv\>{(\IF i=j\> \THEN \vec{x} \ELSE [])}\> ((\IF i=j \THEN s \ELSE \SSkip); \SBar_{\idx,\TT, \SSin}) \\
    }
    {
    \cdt{SSingle}_\mathit{idx}\> \vec{x}\> s \conc \lenv \conc m  \xmapsto[i]{\fnt{sync_ectx}(\SinCtx_\mathit{idx}\ j) }  s' \conc \lenv \conc m
    }

  \end{mathpar}
}

\caption{Semantic rules: \kw{parallel}, \kw{for}, \kw{single}, \kw{barrier}}
\label{fig:semantics-2}
\vspace{-0.5cm}
\end{figure}
\end{DIFnomarkup}

\newpage
\subsection{Parallelism and Work-Sharing}
\label{sec:parallelism-and-work-sharing}

\begin{figure}[htb]
{
\vspace{-0.5cm}
\small
\begin{minipage}{.5\textwidth}
\begin{lstlisting}[style=customc]
#pragma omp parallel num_threads(5)
{
#pragma omp for
    for(int i = 0; i < 100; i++) {
        printf("%d\n", i);
    }
}
\end{lstlisting}

\caption{An OpenMP program with \phantom{blah} \kw{parallel} and \kw{for} constructs.}
\label{fig:ex-for-prog}
\end{minipage}
%
\begin{minipage}{.03\textwidth}
\phantom{a}
\end{minipage}
\begin{minipage}{.5\textwidth}
\begin{lstlisting}[style=customc, showlines=true]
#pragma omp parallel num_threads(2)
{
#pragma omp single
    { printf("Before\n"); }
    printf("After\n");
}

\end{lstlisting}
\caption{An OpenMP program with a \kw{single} construct.}
\label{fig:single-prog}
\end{minipage}
\vspace{-0.9cm}
}
\end{figure}

\subsubsection{The Parallel Construct}

The main structuring construct in OpenMP is the \kw{parallel} construct, which creates a team of threads to execute a parallel region. When a thread $i$ reaches a \cdt{SPar} statement, the rule \textsc{Step-Parallel} forks a team of threads executing the region's body $s$ in parallel. For example, in \Cref{fig:ex-for-prog}, a team of 5 threads executes lines 2--6 in parallel. Thread $i$ also participates in the team and becomes the \emph{primary thread} of the team, so we only need to add $\mathit{nc}-1$ new threads. We add these threads both to the team tree \ttree via the \fnt{spawn\_team} operation and to the thread pool \tp, where we record each thread's code, local environment, and permissions. At a first pass, the code executed by the spawned threads is $(s; \kw{SBarrier}_{\TT,\kw{SPar}})$: each thread executes the body $s$ of the parallel region, and then a barrier synchronizes all the threads in the team and ends the parallel region. 

For the permissions, all teammates split thread $i$'s original permissions arbitrarily---a given memory location may be written by a single thread in the team or read by all of them (but not both). In this case, the original permission $\pi_i$ is redistributed among the team ($\pi_i'$ and $\pi_j'$ for $j\IN \vec{t}$).

Two points further complicate this picture: \emph{privatization} and \emph{reduction}. Each thread wraps its code in a \kw{SPriv} command that creates private copies of variables as specified in the \kw{SPar} directive's \kw{private} clause. Furthermore, if the directive's \kw{reduction} clause is non-empty, the threads must execute additional commands (denoted by \kw{SBRB}) to perform reduction at the end of the region. We discuss privatization and reduction in \Cref{sec:variables}.




\subsubsection{The For Construct}
\label{sec:for}

The \kw{for} construct is the most common form of work-sharing in OpenMP, distributing the iterations of a for loop across the threads in a team. 
For example, the program in \Cref{fig:ex-for-prog} distributes its loop to a team of 5 threads, each of which executes some portion of the loop in parallel. The program will still print each of 0 to 99 once, but in an arbitrary order\footnote{This loop body is safe to execute in parallel because the iteration variable \code{i} is made \emph{private}; we explain this in more detail in the following sections.}
 determined by the partition of iterations and the order of parallel execution.



The rule \textsc{Step-For} applies when a thread reaches a \kw{SFor} pragma applied to the loop statement $s$. The OpenMP standard requires that $s$ be in \emph{canonical loop nest form}\footnote{Our semantics formalizes requirements about the initializer, the condition, and the increment statements in the loop. The OpenMP specification imposes further requirements on the loop body: there must be no \code{break} or \code{continue} statements, and the iteration variable must not be modified during execution of the loop body, preventing partial or skipped execution of some iterations. Since this is a syntactic constraint, we assume that a parser has already implemented this check when it generates the ClightOMP language.},
essentially restricting the syntax of the loop so that it is possible to compute the \emph{logical iterations} at runtime by just looking at the initializer, increment, and condition statements of the associated loop. We then choose a partition $w$ that maps thread ID to a list of iterations, so $w(i)$ is the iterations for thread $i$. The partition $w$ is arbitrary as long as these iterations combine to a permutation of the original iterations. The team must agree on a context $\ForCtx$ for the loop that includes $w$, guaranteeing that each thread uses the same partition. This is accomplished using the $\fnt{sync_ectx}(\ForCtx_\idx\ w)$ operation: if there is no \ectx\ for its team in \ttree, it is set to be $\ForCtx_\idx w$; otherwise, if there is already some \ectx\ set, $\fnt{sync_ectx}$ is only defined if $\ectx=\ForCtx_\idx\ w$, essentially enforcing the team to agree on the same $w$.
The thread then executes a modified version $s'$ of the loop, where the iterations have been restricted to its piece $w(i)$ by modifying the initialization and test expressions, followed by a closing barrier. As with \kw{parallel}, any privatization and reduction clauses in the \kw{for} pragma are handled by \kw{SPriv} and \kw{SBRB} respectively.





\subsubsection{The Single Construct}
\label{sec:single}

We also model the simpler work-sharing construct \kw{single}, which says that only one thread in the team executes the associated code. For example, in \Cref{fig:single-prog}, exactly one thread prints ``Before'', and 2 threads print ``After'', and ``Before'' is always printed before any ``After''.
As with \kw{for}, the threads must agree on a $\SinCtx$ by $\fnt{sync_ectx}$, but this time its only contents are the id $j$ of the thread that will execute the code. Threads that are not chosen execute no instructions (\SSkip) and privatize no variables ($[]$) in the region. The \SSin\ construct does not support reduction clauses, since only one thread performs computation in the region.

\subsubsection{Barriers}
\label{sec:barrier}

\begin{figure}[htb]
{\small
\vspace{-0.5cm}
\begin{minipage}{.5\textwidth}

\begin{lstlisting}[style=customc]
#pragma omp parallel num_threads(5)
{
    printf("before\n");
#pragma omp barrier
    printf("after\n");
}
\end{lstlisting}
\subcaption{threads wait at a \kw{barrier}\\\hspace{\textwidth}}
\label{fig:barrier-ex}
\end{minipage}
\begin{minipage}{.03\textwidth}
\phantom{blah}
\end{minipage}
\begin{minipage}{.5\textwidth}
    \begin{lstlisting}[style=customc]
#pragma omp parallel num_threads(2)
{
    int i = omp_get_thread_num();
    if (i==0) { bar1(); } 
    else { bar2(); }  
}
\end{lstlisting}
\subcaption{Threads in a team should enter \kw{barrier}s in the same order}
\label{fig:two-barrier-ex}
\end{minipage}
}

\caption{OpenMP programs with the \kw{barrier} constructs.}
\end{figure}

The \kw{barrier} construct creates an explicit barrier that blocks a thread until all threads in its team reach the barrier. For example, in \Cref{fig:barrier-ex}, all 5 threads have to reach the barrier before any of them can proceed, so all ``before''s will be printed before any ``after''.

When thread $i$ takes a \textsc{Step-Barrier}, it first identifies its team members:
\begin{defi}
\emph{The team of a thread $i$ in a team tree \ttree\ is the set of thread IDs of all the siblings of $\ttree(i)$ (including $i$), denoted as $\TeamTIDS (\ttree, i)$. }\footnote{Notice a sibling node of $\ttree(i)$ is not necessarily active since it may spawn another team, but only the thread of the sibling node is in $\TeamTIDS(\ttree, i)$.  }
\end{defi}
\textsc{Step-Barrier} requires all threads in this team wait at the \emph{same} barrier, and the \textit{idx} of the barrier distinguishes syntactically distinct barriers. Consider the program in \Cref{fig:two-barrier-ex}: \lstinline[style=customc]{omp_get_thread_num} \footnote{ An OpenMP thread number uniquely identifies a thread in a team. For thread $i$, the function \lstinline[style=customc]{omp_get_thread_num} returns the index of node $\ttree(i)$ relative to its siblings. Note that this is different from a thread ID, which uniquely identifies a thread in the thread pool. } returns 0 for the leader and 1 for the other thread, so two threads in a team arrive at different barriers and should be stuck according to the OpenMP specification. \footnote{This program terminates for OpenMP implementations in GCC or Clang, and to the best of our knowledge, there are currently no tools that catch this issue. }

Our semantics gives uniform treatment to both explicit \kw{barrier} directives and the implicit barriers at the ends of regions (i.e., the barriers in {\SBRB} generated by rules like \StepPar), and we index $\kw{SBarrier}_{\textit{idx,b,s}}$ with not just $\idx$ but an additional boolean $b$ and an optional statement $s$ (the directive that generates the implicit barrier, or $\bot$ for an explicit barrier)  to handle the subtle differences in their semantics. As shown in the rule \textsc{Step-Barrier}, a barrier can only execute when a full team of threads are all at the same barrier (i.e., a barrier with the same indices \textit{idx,b,s}). Explicit barriers must not be placed inside team-executed regions, which we check via the team's \fnt{ectx}; this requirement comes from the OpenMP specification and serves to rule out, e.g., the case where a barrier in the body of a shared loop is encountered a different number of times by threads executing a different number of iterations of the loop body. 
Further, if the barrier is the closing operation ($b = \mathsf{true}$) for some construct $s$, it removes related context from \ttree\ accordingly with $\fnt{end_ctx}$: for a \kw{for} construct ($s = \SFor$) or a \kw{single} construct ($s=\SSin$), it removes the \fnt{ectx}; for a \kw{parallel} construct ($s=\SPar$), it removes the team (i.e. undo the effects of $\fnt{spawn_team}$).

Finally, all threads in the team simultaneously move to the next statement, and may also freely exchange permissions within the team: for instance, a barrier may be used to separate a code segment in which a thread has exclusive {\Writable} access to a variable from one in which {\Readable} permission is shared across the whole team.


\subsection{Privatization and Reduction}
\label{sec:variables}
As described in \Cref{sec:motivationexample}, OpenMP directives can have subtle effects on the semantics of variables within the affected regions. In particular, both \kw{parallel} and \kw{for} constructs can trigger \emph{privatization} and \emph{reduction}, according to the private clause $\mathit{pc}$ and the reduction clauses $\mathit{rcs}$ included in their directives.  We model the semantics of these clauses as additional helper statements emitted at the start of each region, with arguments derived from both the syntactic clauses and the program's runtime state. 
The rules are defined in \Cref{fig:semantics-3}.

\begin{DIFnomarkup}
\begin{figure*}[ht]
{\footnotesize \centering
\vspace*{-0.5cm}
\hspace*{-1.5cm}
\begin{minipage}{1.2\textwidth}
\begin{mathpar}
    
\inferrule[Step-Priv]
    {
        \fnt{priv}_{\vec{x},i}(\lenv, m)=\tuple{\lenv', m'} \\
    }
    {
      \SPriv \> {\vec{x}}\> s \conc \lenv \conc m \XMAPSTO{i} s;\SPrivEnd \> \lenv|_{\vec{x}} \conc \lenv' \conc m'
    }

\inferrule[Step-Priv-End]
    {
        \fnt{end_priv}_{i, \lenv_o} (\lenv, m)=\tuple{\lenv', m'} \\
    }
    {
        \SPrivEnd \> {\lenv_o} \conc \lenv \conc m \XMAPSTO{i}  \SSkip \conc \lenv' \conc m'
    }

\inferrule[Step-Red]
    {
        r = \tuple{\vrc, \lenv_o} \\
        \vec{t} \triangleq \fnt{team_tid}(\ttree, i) \\
    }
    {
        \SRed\> r \conc \lenv \conc m \XMAPSTO{i, \tp, \ttree} \SSkip \conc \lenv \conc  
        \IF \fnt{leader}(\ttree, i) = i \THEN \fnt{red}_{\vec{t}, \vrc, \lenv_o, \tp}(m) \ELSE m
    }

\end{mathpar}
\end{minipage}
}
\caption{Semantic rules: \kw{private}, \kw{reduction}. Supporting operations $\fnt{priv}$, $\fnt{end_priv}$ and $\fnt{red}$ are defined in \Cref{fig:semantics-pr} in the appendix}
\label{fig:semantics-3}

\vspace{-0.5cm}

\end{figure*}
\end{DIFnomarkup}

\subsubsection{Privatization}
\label{sec:privatization}
\begin{figure}[t]
{\small

\begin{lstlisting}[style=customc]
int x = 3;
#pragma omp parallel private(x) num_threads(2)
{ // x uninitialized
    x = omp_get_thread_num();
    printf("my thread number is %d\n", x);
} // x=3
\end{lstlisting}
}
\caption{An OpenMP program with a \code{private} clause, adapted from \cite{mattsonIntro}.}
\label{fig:private-ex}

\vspace{-0.3cm}

\end{figure}

By default, variables declared before a parallel region are \emph{shared} among the threads in the region's team: they are treated as memory locations that can be accessed by all threads (and access must be synchronized accordingly). However, sometimes it is useful for each thread in the team to have its own \emph{private} copy of a shared variable. In the example in \Cref{fig:private-ex}, if \code{x} were a shared variable, the threads' accesses to it would race; the intended behavior is instead that each thread uses its own copy of \code{x}, as indicated by the \code{private} clause in the \kw{parallel} directive. Both \kw{parallel} and \kw{for} constructs can include \code{private} clauses; the iteration variable of a \kw{for} construct is also implicitly privatized as mentioned in \Cref{sec:for} (we assume that the parser has already added any implicitly privatized variables to the relevant \code{private} clauses).

In \Cref{fig:private-ex}, if there were no \code{private} clause, variable \code{x} would be shared in the team in the parallel region, and updates to \code{x} on line 5 would race. With the \code{private} clause, a private copy of \code{x} is allocated for each thread in the team, so it is safe to update \code{x} on line 5, and on line 6 each thread's \code{x} holds its own OpenMP thread number. After the parallel region, these private copies of \code{x} are deallocated, and the original \code{x} is again in scope.

The execution of a region with a \kw{private} clause is prefixed with a statement $\kw{SPriv}\ \vec{x}$, where $\vec{x}$ is the list of variables to be privatized. The rule \textsc{Step-Priv} describes the semantics of privatization. For each private variable $x_k \IN \vec{x}$, we allocate a new location $l_k$ in the shared memory $m$ which will hold the value of thread $i$'s private copy of $x_k$, and set $x_k$ to refer to $l_k$ in $i$'s local environment. The initial value at $l_k$ is undefined, as if it were a newly declared variable (with the exception that variables in reduction clauses are implicitly privatized and initialized\footnote{OpenMP has rather detailed rules for determining the initial values of reduction variables, which we do not discuss here but are reflected in our semantics.}). 
We then add an \kw{SPrivEnd} instruction to the end of the region, with an argument $\lenv_o$ containing the original memory locations of the variables in $\vec{x}$. When we reach the \kw{SPrivEnd} statement, we reverse the process, freeing each location allocated for a private variable and restoring the local environment to the original locations of $\vec{x}$. Note that each thread in a team will separately execute the \kw{SPriv} instruction, allocating its own copies of the variables in $\vec{x}$; after all threads end privatization, each $x_k$ will once again refer to the same memory location in each thread.


\subsubsection{Reduction}
\label{sec:reduction}
\begin{figure}[htb]
{\small

\vspace{-1cm}

\begin{lstlisting}[style=customc]
int sum=0;
#pragma omp parallel num_threads(2)
{
#pragma omp for reduction(+:sum)   
    for(int i=1; i<=100; i++) { sum += i; }
}
printf("sum of the first hundred numbers is %d\n", sum);
\end{lstlisting}
}
\caption{An OpenMP program with a \code{reduction} clause}
\label{fig:reduction-ex}

\vspace{-0.3cm}

\end{figure}
A reduction variable is a special kind of private variable for which, instead of returning to the original value after the region ends, we apply some operation to combine all of the private copies into a new value for the shared variable. For example, in \Cref{fig:reduction-ex}, the loop iterations are divided between two threads, each summing a part of the 100 numbers. Each thread saves its running sum in an implicitly privatized \code{sum} variable during the \kw{for} construct. When both threads are at the end of the \kw{for} construct, the reduction clause says that the private copies of \code{sum} are combined with the original one with the operator \code{+}, producing the total sum.

We assume that all reduction variables are declared to be private (the parser can add all reduction variables to the \code{private} clause of the same construct), so they are already processed by the \kw{SPriv} statement for their region. The additional work of reduction occurs at the end of the region: if the region's \code{reduction} clause is non-empty, then instead of a single barrier, \kw{SBRB} emits a barrier, a reduction statement \kw{SRed}, and another barrier\footnote{We do not support the \code{no_wait} clause that skips these implicit barriers.}. The first barrier synchronizes all threads in the team, ensuring that they have all finished the region, and transfers permissions for all their reduction variables to the leader. The \kw{SRed} statement is executed only by the leader and performs reduction on all private copies of each reduction variable, folding the specified operation over their values and storing the result in the original variable. (For all non-leader threads, \kw{SRed} is a no-op.) Finally, the last barrier forces all threads to wait until the reduction is complete before proceeding to any remaining code, and it also redistributes the team's permissions.





\subsection{Design Considerations}
There are many different ways of translating the OpenMP specification document into operational semantics. Our semantics is aimed to uncover variable and synchronization errors, and so the following considerations are especially important:
\paragraph{Barrier is the only operation that synchronizes running threads.} An operational semantics is easiest to reason with when each step is performed by one thread, and concurrency can be expressed by the interleaving of steps of individual threads. Every one of our rules except \textsc{Step-Barrier} is local in this sense: each thread may enter a \kw{for} or \kw{single} region, or execute privatization or reduction, independently of all other threads, with coordination accomplished via the asynchronous \fnt{sync\_ectx} (the first thread to enter a region sets the \fnt{ectx}, and the others read it). Only \textsc{Step-Barrier} requires that all threads in a team be at the same instruction and moves all threads forward simultaneously. Thus, all synchronization in an execution can be traced specifically to barrier operations (or to thread creation via \textsc{Step-Parallel}).
\paragraph{Operations that can race are not atomic.} If two operations can occur in either order, it is important that our semantics observe the resulting nondeterminism, which is hidden if the operations are both performed as part of a single step. For example, we could have combined \textsc{Step-End-Priv} and \textsc{Step-Red-Leader} into a single end-of-region step, but since both access privatized variables, this would obscure possible races between reduction and ending privatization. By separating them, our semantics identifies the need for a barrier between the two steps\footnote{This synchronization is not explicitly mentioned in the OpenMP specification, but its absence leads to races in programs that should be well-defined, and both GCC and Clang place a barrier before reduction.}, and could identify possible races if, e.g., we implemented the \kw{nowait} clause that removes the barrier at the end of a region. 
Another consequence of this separation is that there is no step that both synchronizes threads and accesses memory, which leads to the next point.
\paragraph{A thread must have permission to perform its memory accesses.} Steps that access memory (such as privatization and reduction) only succeed if the thread performing them has the required permissions at the start of the step. This allows the permissions to serve their intended purpose of race detection: the fact that permissions are always consistent in an execution state is enough to guarantee the absence of data races. We discuss race-freedom further in \Cref{sec:drf}.

As a side note, we are also able to support the OpenMP \kw{critical} construct by elaborating it to \kw{acquire} and \kw{release} of a dedicated lock, since lock operations are already supported by the CPM.

\section{Properties of the ClightOMP Semantics}
\label{sec:property}
The ClightOMP semantics inherit from Clight and the CPM some safety properties that are implicitly proved as part of demonstrating an execution. Moreover, demonstrating an execution also guarantees that the execution is free of data races.

\subsection{Invariants and Safety Properties}

First, Clight's semantics enforces memory safety:
\begin{theorem}[from CompCert] If a thread uses the value of a variable, then that variable has a defined value. If a thread step accesses memory, then that memory access is valid.
\end{theorem}
\noindent So proving that a program (even a sequential one) can execute under our semantics implies that threads never use uninitialized variables, access unallocated or deallocated memory, access arrays out of bounds, etc.

Second, the CPM enforces \emph{permission coherence} before each step: each thread's permissions must be consistent with each other thread's, i.e., no two threads can have write access to the same location at the same time. 
\begin{theorem}[from CPM] If an execution from the initial state of a program reaches a state $\langle \mho, \tp,  \ttree, m\rangle$, then the permissions of threads in $\tp$ are coherent.\end{theorem}
\noindent This is most relevant to our semantics when we perform the \fnt{sync} operation: we may redistribute permissions freely within the team at team creation and at barriers, but the new distribution must always be coherent. This is useful for proving race freedom, as we will see in the next section.

As a consequence of these properties, demonstrating an execution of a program under our semantics implies the absence of several categories of common bugs. Variable initialization, memory errors, and data races are responsible for many real-world errors in OpenMP programs~\cite{openmp-diff-test}, so even without proving specific properties of a program (e.g., functional correctness), showing that it executes according to ClightOMP makes it much more likely to be correct.

\subsection{Race-Freedom}
\label{sec:drf}
The CPM's permissions are intended to guarantee race-freedom of any execution under its semantics, and its authors prove that at the x86-TSO level, successful executions contain no data races~\cite{cpm}. Our semantics enjoys a similar property at the C level, using the fact that every thread only performs memory operations that are allowed by its current permissions:

\begin{theorem}[Data race freedom] If an execution contains two conflicting events $e_1$ and $e_2$ on the same memory location, then $e_1$ and $e_2$ are synchronized.
\end{theorem}
\noindent The proof is in \Cref{sec:drf-proof}. Thus, any successfully terminating execution in our semantics has not encountered undefined behavior. This also justifies our use of a sequentially consistent memory model, since OpenMP specifies that race-free programs without weak-memory \kw{atomic} directives (which we do not support) have sequentially consistent behavior.

%% file: comparison.tex
\section{Related Work}\label{sec:comparison}

\myparagraph{C Semantics and Concurrency}
Our work builds on the Concurrent Permission Machine~\cite{cpm}, which is part of an effort to modularly lift CompCert's semantics and correctness proof to concurrency. Other approaches to adding concurrency to CompCert include CompCertTSO~\cite{compCertTSO}, which directly extends CompCert with the TSO memory model, and CASCompCert~\cite{cascompcert}, which is similar to the CPM but works with memory footprints instead of permissions. CASCompCert's semantics are simpler than the CPM's and do not involve as much explicit permission management, but it requires that concurrency primitives be implementable as atomically executed blocks of ordinary C commands, which does not permit the team management involved in OpenMP directives.

Two other major formalizations of C's semantics are KC~\cite{ellison-rosu-2012-popl} and Cerberus~\cite{10.1145/2980983.2908081,Memarian:2019:ECS:3302515.3290380}. The former does not include concurrency at all, but the latter does, and could also be extended with OpenMP support. While CompCert/the CPM has the advantage of connecting to a verified compiler, Cerberus has an associated model checking tool that is useful for detecting unusual concurrent behaviors.

\myparagraph{OpenMP Semantics and Verification}
Atzeni and Gopalakrishnan~\cite{atzeni} previously gave semantics to OpenMP for the purposes of defining a sound race detector. They also organize the thread hierarchy in a tree-like data structure similar to our team tree, with a different mechanism, \emph{offset-span label}, to describe specific positions within the structure. However, their work is focused entirely on race detection: it elides the semantics of actual source programs, and models program execution as a stream of memory operations and concurrency events (e.g., \kw{parallel} and \kw{barrier} directives). Thus, there are several classes of OpenMP errors that are invisible to their semantics, including malformed control flow/region nesting and reading from uninitialized private variables. In fact, even the \kw{for} construct is not explicitly modeled in their semantics, since its effects are entirely at the level of variables and control flow.

The CIVL concurrency verifer~\cite{civl} supports verification of OpenMP programs, as well as other concurrency frameworks such as MPI and CUDA. It does so by translating OpenMP directives into concurrency primitives in CIVL-C, a custom C extension with built-in concurrency operations. While CIVL-C has formal semantics and an associated verifier via SMT solvers, the translation of OpenMP to CIVL-C is part of CIVL's trusted computing base. Our work provides a means to verify this translation by showing that the CIVL-C implementation of OpenMP is consistent with its source-level semantics.

\ignore{
Our work is closely related to prior formal semantics of C and to C-like concurrent memory models.

\myparagraph{Formalizing C}
A number of prior works have formalized the semantics of C,
including CompCert~\cite{Blazy2009,leroy:hal-00703441},
KC \cite{ellison-rosu-2012-popl}, \citet{Kang:2015:FCM:2813885.2738005},
and \citet{10.1145/2980983.2908081,Memarian:2019:ECS:3302515.3290380}. 
Li \textit{et al.} proposed a formalism for Checked-C~\cite{checkedc}, a memory-safe dialect of C. \wm{is this relevant?}

\myparagraph{OpenMP and Related Memory Models}
As mentioned in \Cref{sec:intros}, there have been several previous attempts to model OpenMP concurrency \cite{atzeni,Gazi2024}. 
For the other concurrency model, Lamport probably was the first to define a memory model weaker than sequential consistency for multi-threaded programs \cite{Lamport:1979:MMC:1311099.1311750}. Adve and Hill \cite{Adve:1990:WON:325096.325100} started defining weak memory orders for memory operations.
Focusing just on hardware models: Ahamad \textsf{et al.} \cite{Ahamad1995} axiomatized causal memory and proved some important theorems. Higham \textsf{et al.} \cite{Higham98weakmemory} formalized SPARC and a number of simpler memory models in both axiomatic and operational styles. Sevcik \textsf{et al.} created a formal verification framework for a small C-like language \cite{relaxed-memory}. The same group \cite{compCertTSO} later developed the CompCertTSO to verify a compiler from CLight to X86 based on a relaxed memory model. 
Lochbihler verified a whole-program compiler for multi-threaded Java \cite{Lochbihler2018}. Sevcik \textsf{et al.} built \mbox{CompCertTSO} \cite{compCertTSO}, which adapted CompCert’s correctness proofs to x86TSO in order to consider the compilation of racy C code. 

Alglave \textsf{et al.} \cite{Alglave:2014:HCM:2633904.2627752} specified in great detail how to use a candidate execution model to define relaxed memory models and provided several verification tools. The C11 memory model was designed by the C++ standards committee based on a paper by Boehm and Adve \cite{c++memoryarticle}. Batty \textsf{et al.} formalized the C11 model with some improvements and proved the soundness of its compilation to X86-TSO \cite{Batty:2011:MCC:1925844.1926394}. A number of papers \cite{DBLP:journals/tinytocs/DoddsBG13,Vafeiadis:2013:RSL:2544173.2509532,Norris1,Boehm:2014:OGA:2618128.2618134} found that Batty \textsf{et al.}'s model enabled OOTA behaviors. Vafeiadis \textsf{et al.} \cite{CompilerOptimisations} found many other problems in Batty \textsf{et al.}'s model and proposed fixes. In 2016, 
 Batty \textsf{et al.} proposed a more concise model for $\texttt{sc}$ atomics \cite{Batty:2016:OSA:2914770.2837637}, but the model is stronger than C11; and the $\texttt{sc}$ fences there are too weak. Much previous work \cite{Vafeiadis:2013:RSL:2544173.2509532,Meshman:2015:PSS:2893529.2893552,Lahav:2015:ORW:2958675.2958702,Tamingreleaseacquire} focused on a fragment of C++ concurrency. In 2017, Lahav \textsf{et al.} \cite{Lahav:article} defined a comprehensive C++ model (RC11) based on all previous models, with extra fixes on Batty \textsf{et al.}'s model.  Many previous papers \cite{ConcurrencySemanticsRelaxed,Jeffrey:2016:TAR:2933575.2934536,Relaxedmemory} also proposed solutions for OOTA problems. \wm{I'm not convinced that any of this is at all relevant to us. We're not talking about weak memory in this paper at all. We should compare to OpenMP semantics, operational concurrent C models (e.g. CASCompCert), and maybe concurrent semantics for other languages.}
}

%% file: conclusion.tex
\section{Conclusion and Future Work}
\label{sec:conclusion}

In this work, we have presented ClightOMP, a formal semantics for C programs with OpenMP directives. Our semantics extends the concurrent C semantics of the Concurrent Permission Machine with a \emph{team tree} that tracks the states of threads spawned via OpenMP, and carefully accounts for the effects of privatization and reduction on local variables. This gives us a formal definition of the allowed behaviors of OpenMP C programs according to the OpenMP specification, and enables us to identify subtle, undesirable behaviors induced by incorrect annotations; in particular, we have shown that any execution in our semantics is data-race-free. Because most of the OpenMP logic is captured by changes to the team tree, our semantics are also fairly modular and could be reused, with minimal modification, to define OpenMP semantics for other languages (e.g., C++ or Fortran).

We aim to further validate the consistency of ClightOMP semantics with the behaviors of actual OpenMP programs in two ways. First, we can test the semantics by proving that some executions of OpenMP programs have the same results in our semantics as their real-world implementations; second, we can implement a verified compiler of ClightOMP programs and test the compiled programs. 

We envision two main applications for this semantics: \emph{testing} and \emph{verification}. On the testing side, we aim to develop an \emph{executable} version of our semantics that could serve as a reference interpreter for OpenMP programs. An executable semantics would allow us to test our semantics against actual OpenMP implementations, potentially uncovering bugs in either our semantics or the implementations. It could also serve as a basis for property-based testing (PBT), automatically verifying the correctness properties of OpenMP programs in a sound manner grounded in our semantics. The PBT framework would be useful for detecting race conditions and runtime errors in OpenMP programs.

On the verification side, tools that automatically parallelize sequential programs by instrumentation with OpenMP pragmas \cite{tehrani2024coderosetta,mahmud2025autoparllm,mahmud2025contraph} lack a semantic preservation guarantee. To aid this, we can prove that a program is correctly instrumented with OpenMP directives by proving that its allowed behaviors under our semantics are a subset of those of the original sequential program (i.e., refinement). We are also interested in using our semantics to verify a \emph{compiler} that implements OpenMP directives with standard concurrency primitives: the CPM already has a (mostly) proved-correct compiler to assembly, so by translating OpenMP directives into CPM spawn and lock operations, we can obtain a verified compiler for C+OpenMP. In the longer run, we note that the CPM was developed as backing for the Verified Sofware Toolchain (VST)~\cite{vst}, a tool for verifying functional correctness of C programs; our semantics could serve as the basis for an extension of VST to verify OpenMP programs as well, including those with more complex behavior that may not exactly match the behavior of sequential programs (at least locally).



%% file: syntax.tex
\begin{figure}
  \small \centering
  $\begin{array}{l}
\begin{array}{lll}
  \text{Identifier:}~ x
& \text{Integers:}~n::=\mathbb{Z} 
& \text{Booleans:}~b::=\TT \mid \FF
\end{array}
\\[0.5em]

\begin{array}{llllllll}
 \mathit{rid} & ::= & \texttt{+} \mid \texttt{*} \mid \& \mid \texttt{|} \mid \mathtt{\sim} & \text{reduction\_identifier}  \\[0.5em]
                             & \mid & \texttt{\&\&} \mid \texttt{||} \mid \texttt{max} \mid \texttt{min} \\[0.5em]
nc & ::= & n & \text{num\_thread clause} \\[0.5em]
pc & ::= &  x^* & \text{privatization clause} & \verb|private(x*)| \\ [0.5em]
rc & ::= & (\mathit{rid}, x^*) & \text{reduction clause} & \verb|reduction(rid:x*)| \\[0.5em]

r  & ::= & (rc^*, \lenv) & \text{data for reduction} \\[0.5em]
\mathit{s} & ::=   & \mathit{cs} & \text{Clight Stmt}  & \\[0.5em]
& \mid & \kw{SPar}_{n} \> nc \> pc \> rc^* \> s & \text{parallel construct} & \verb|#pragma omp parallel| \\[0.5em]
& \mid & \kw{SFor}_{n} \> pc \> rc^* \> s & \text{for construct} & \verb|#pragma omp for| \\ [0.5em]
& \mid & \kw{SSingle}_{n} \> pc \> s & \text{single construct} & \verb|#pragma omp single| \\ [0.5em]
& \mid & \kw{SBarrier}_{n, b, s} & \text{barrier construct} & \verb|#pragma omp barrier| \\ [0.5em]
& \mid & \kw{SPriv} \>\mathit{pc} \>s  & \text{start a privatization scope} \\ [0.5em]
& \mid & \kw{SPrivEnd} \> \lenv  & \text{end of a privatization scope} \\ [0.5em]
& \mid & \kw{SRed} \> r  & \text{reduction} \\ [0.5em]
\end{array}
    \end{array}
  $
  \caption{\lang Syntax. \kw{SPriv}, \kw{SPrivEnd}, \kw{SRed} are generated at runtime, and \ensuremath{\lenv} stores runtime values. }
  \label{fig:ClightO-syn}
\end{figure}

%% file: nested_par_region.tex
\Cref{fig:nested-parallel-region-diagram} depicts the team tree states in a possible execution of the program from \Cref{fig:program-nested-par-region}.
The program starts with a single thread $\mathit{t1}$, just like a sequential C program, and the corresponding initial team tree, as depicted in \Cref{fig:nested-parallel-region-diagram}(a). The subscript $i$ of an OpenMP thread state denotes the thread ID, and the superscript indicates the level of nesting in a parallel region, or \textit{init} to specify a freshly initialized node. 

When thread $\mathit{t1}$ reaches the first \kw{parallel} construct on line 2, it starts a new team that executes the code associated with this construct (lines 3--9).

As depicted in \Cref{fig:nested-parallel-region-diagram}(b), the \kw{parallel} construct on line 2 starts a new parallel region, triggering \fnt{spawn_team} that adds a new team of two nodes of thread IDs $\mathit{t1}, \mathit{t2}$ (along with some team context that is not shown here). The new node $s_{t1}^2$ supersedes the original node $s_{t1}^1$ for thread $\mathit{t1}$, and starts with its own initial (empty) stack of OpenMP contexts, representing the fact that ``$\mathit{t1}$ as part of the team $\{\mathit{t1}, \mathit{t2}\}$'' is a different logical entity from ``$\mathit{t1}$ as the initial thread''. When the team ends at line 9, although $\mathit{t2}$ will terminate and $\mathit{t1}$ will continue to execute,
both nodes in the team are removed from the tree by \fnt{despawn_team}; $\mathit{t1}$ will return to OpenMP state $s_{t1}^{1}$, its state from before the \kw{parallel} construct (\Cref{fig:nested-parallel-region-diagram}(d)).

Nested \kw{parallel} constructs create nested parallel regions. As depicted in \Cref{fig:nested-parallel-region-diagram}(c), when $\mathit{t1}$ of the team $\{\mathit{t1}, \mathit{t2}\}$ meets the parallel construct on line 4 in \Cref{fig:program-nested-par-region}, it spawns a new team with the same rules, while $\mathit{t2}$ has not reached line 4 yet. When $\mathit{t2}$ reaches line 4, it likewise spawns a new team $\{\mathit{t2}, \mathit{t5}, \mathit{t6}\}$ for parallel region III for some new threads $\mathit{t5}$, $\mathit{t6}$, and the execution of this team does not interfere with the other team $\{\mathit{t1}, \mathit{t3}, \mathit{t4}\}$.

\begin{figure}[h]

\begin{tikzpicture}

{[on background layer]
\draw[rounded corners, dotted, blue, dash pattern=on 2pt off 2pt] 
    (0.5, -0.1) 
    rectangle 
    (12.8, -4.2);

\draw[rounded corners, dotted, blue, dash pattern=on 2pt off 2pt] 
    (0.7, -0.6) 
    rectangle 
    (12.6, -3.7);

\draw[rounded corners, dotted, blue, dash pattern=on 2pt off 2pt] 
    (0.9, -1.4) 
    rectangle 
    (12.4, -3.3);

}

\node[anchor=north west, inner sep=0] (codebox) at (0,0) {%
    \begin{minipage}{1.2\linewidth}
{\small
{\captionsetup[lstlisting]{margin = 4 mm}
  \begin{lstlisting}[xleftmargin= 4 mm]
    // program starts with one thread t1 (Parallel Region I)
    #pragma omp parallel num_threads(2) // (Parallel Region II)
    { // a team of 2 threads {t1, t2} is created
    #pragma omp parallel num_threads(3) //(Parallel Region III)
        { // t1 creates a team of {t1, t3, t4}, t2 creates {t2, t5, t6}
            do_something
        } // the teams {t1, t3, t4},  {t2, t5, t6} ends here, the primary 
          //  thread of each team t1 and t2 resumes to parallel region II  
    } // the team {t1,t2} ends
    // t1 resumes in the outmost parallel region I
  \end{lstlisting}
}
}

 \end{minipage}
};

\end{tikzpicture}

\caption{An OpenMP program with nested parallel regions I, II and III, marked with dotted squares.  }
\label{fig:program-nested-par-region}
\end{figure}

\begin{figure}[htb]
  \includegraphics[width=1.06\textwidth]{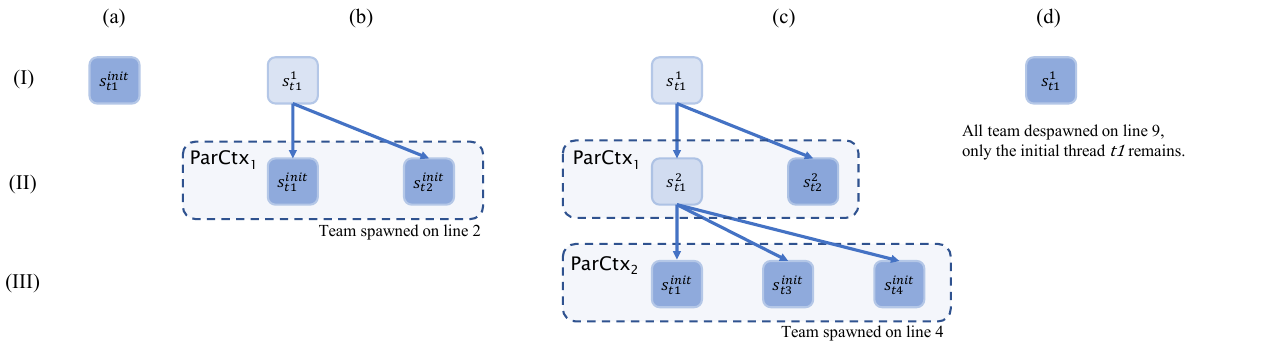}
\vspace*{-0.5em}
  \caption{Team tree \ttree evolution for an execution of the program in \Cref{fig:program-nested-par-region}. \textcolor{spec}{Blue} nodes are active; \textcolor{auto}{grayed out} ones are inactive. (I), (II), and (III) refer to the regions in \Cref{fig:program-nested-par-region}. Each team also has an associated parallel context. }
\label{fig:nested-parallel-region-diagram}
\end{figure}

%% file: ttree_op.tex
\vspace{-1cm}
\begin{DIFnomarkup}
\begin{figure}[H]
{
Node operations of the form $\fnt{op}\triangleq \tuple{r, f}$ is a pair of two functions $r, f:\mathit{Node} \rightarrow \mathit{Node}$, in which $f$ updates a node, and $r$ finds the node that needs to be updated with respect to the input node: $\id$ is the identity and $\fnt{parent}$ finds the parent node. $n$ in the $\mathrm{when}$ clause binds to the node before applying $\fnt{op}$, and $\fnt{op}$ is defined only if the $\mathrm{when}$ clause holds.
\[
    \begin{array}{l@{\;}c@{\;}c@{\;}l}
    \fnt{spawn\_team}(\pctx, \vect) &\triangleq\ & \Tuple{\id, &
        \lambda n, n.\fnt{tm} ::=  \tuple{\pctx, \bot, \fnt{new_nodes}(\vect)} \quad \WHEN\; n.\fnt{tm} = \bot } \\
    \fnt{despawn\_team}(\idx) &\triangleq\ & \Tuple{\fnt{parent}, &
        \lambda n, n.\fnt{tm} ::= \bot \quad \WHEN\; n.\fnt{ectx}=\bot\> \land \\ &&&
           n.\fnt{pctx}=\ParCtx_{\idx} \land
           \forall\ n' \IN n.\fnt{mates},\> n'.\fnt{tm}=\bot }  \\
    \fnt{sync_ectx}(\ectx) &\triangleq\ & \Tuple{\fnt{parent}, &
        \lambda n, n.\fnt{ectx} ::= \ectx \quad \WHEN\; n.\fnt{ectx} = \bot \lor n.\fnt{ectx} = \ectx } \\
    \fnt{pop_ectx}() &\triangleq\ & \Tuple{\fnt{parent}, &
        \lambda n, n.\fnt{ectx} ::= \bot \quad \WHEN \;
        n.\fnt{ectx} \neq \bot } \\
    \fnt{end_ctx}(\idx, s) & \triangleq && \IF s=\SFor \lor s=\SSin \THEN \fnt{pop_ectx} \ELSE \fnt{despawn_team}(\idx)
    \end{array}
\]
Tree operations $\fnt{op}(\ttree,i)$ on a tree $T$ and a thread ID $i$, lifted from node operations with the same name. $\fnt{op}$ to the right of $\triangleq$ binds to the node operations:
\[
\begin{array}{l@{\;}l}
\op(T, i) & \triangleq \op.1(\ttree(i)) \upd \op.2 (\op.1(\ttree(i))) \\
& \WHERE \op \IN \{\fnt{spawn\_team}(\pctx, \vect), \fnt{despawn\_team}(\idx), 
\fnt{sync_ectx}(\ectx), \\
&\hspace{1.8cm} \fnt{pop_ectx}(), \fnt{end_ctx}(\idx, s) \} \\
\fnt{leader}(\ttree, i) & \triangleq \text{Thread ID of the $leader$ of the team of $i$} \\
\fnt{team_tids}(\ttree, i) & \triangleq \text{Thread IDs of all teammates in the team of $i$}
\end{array}
\]

}
\caption{Team tree operations $\op(T, i)$, $\fnt{leader}(\ttree, i)$, $\fnt{team_tids}(\ttree, i)$ and the node operations that $\op(T, i)$ is lifted from. }
\label{fig:team-ops}
\end{figure}
\end{DIFnomarkup}

%% file: priv_red_supporting_functions.tex
We first define some notations:
\begin{align*}
\fullmoon_{j}^{[t_0; t_1; \dots; t_k]}\ f_{j} \triangleq & f_{t_k} \circ \dots \circ f_{t_1} \circ f_{t_0} \mathtt{\ where\ } \forall j \in [t_0; t_1; \dots; t_k], f_{j} :  T \rightarrow T \\
d\Bracket{ [k_0 ; \dots ; k_n ] \upd [v_0 ; \dots ; v_n ] } \triangleq & d[k_0 \upd v_0 ] \dots [k_n \upd v_n ] \\
f(\vec{t}) \triangleq & \fnt{map}\ f\ \vec{t} \mathtt{\ where\ } f :  T \rightarrow T, \vec{t} : \mathtt{List}\ T \\
\end{align*}

The privatization operation $\fnt{priv}$ takes the local variable environment and memory right before privatization, allocates private copies of $\vec{x}$ in $m$, and updates $\lenv$. The end of privatization operation $\fnt{end_priv}$ frees variable names in $le_o$, and restores local variable environment for these privatized variables. The reduction operation $\fnt{red}$ gathers contributions in threads $\vec{t}$, adds them to the original variable (by looking up the original local variable environment $\lenv_o$ before privatization).
\vspace{-0.8cm}
\begin{figure}[H]
\begin{DIFnomarkup}
\begin{center}
{\small

\vspace{-2em}

\begingroup
\begin{align*}
    \fnt{priv}_{\vec{x},i}(\lenv, m) \triangleq \ 
    & \letin{ \tuple{m', \vec{l}} }{ m.\fnt{alloc}(\fnt{ctypes}_{\lenv}(\vec{x})) } \\
    & \tuple{\lenv[\vec{x} \upd \vec{l}], m'} \\[5pt]
    \fnt{end_priv}_{i, \lenv_o}(\lenv, m)  \triangleq\ 
    & \fnt{foldr} \big(\lambda \ x.\ \\
    & \hspace{1em}  \letin{ m' } { m.\fnt{free}(\lenv(x)) } \\
    & \hspace{1em} \tuple{\lenv[x \upd \lenv_o(x)],\ m'} \\
    & \hspace{0.2em} \big) \ \ \tuple{le, m}\ \fnt{dom}(\lenv_o) \\  
    \fnt{red\_one\_var}_{\vect, \rop, x, \mathit{v_o}, \tp}(m) \triangleq  \
    & \letin{\vec{v}}{ \fnt{map} \Big(\lambda\ t.\ m\big((\tp(t).\fnt{le})(x)\big) \Big) } \\
    & \letin{v'}{\fnt{foldr}\> \big(\lambda\ a\ b. \fnt{eval_expr} (\kw{EBinOp}\ \rop\ b\ a) \big)\ \ v_o\; \vec{v}} \\
    & m[\mathit{l_o} \upd v'] \\[5pt]
    \fnt{red\_one\_clause}_{\vect, \mathit{rc}, \mathit{le_{o}}, \tp}(m) \triangleq  \
    & \letin{(\rop, \vec{x})}{\mathit{rc}} \\
    & \fullmoon_{x}^{\overrightarrow{xs}} \fnt{red\_one\_var}_{\vect, \rop, x, m(\mathit{le_{o}}(x)), \tp}(m) \\[5pt]
    \fnt{red}_{\vect, \overrightarrow{\mathit{rcs}}, \mathit{le_{o}}, \tp}(m) \triangleq\ & \fullmoon_{\mathit{rc}}^{\overrightarrow{\mathit{rcs}}}\fnt{red\_one\_clause}_{\vect, \mathit{rc}, \mathit{le_{o}}, \tp}(m) \\
\end{align*}
\endgroup






\vspace{-1cm}
}

\caption{Supporting functions for privatization and reduction. $\tp(t).\fnt{le}$\hspace{-1pt} returns the local variable environment of the thread state $\tp(t)$. $\fnt{alloc}$, $\fnt{free}$ allocates and frees a chunk of memory; $\fnt{ctypes}$ looks up the Clight types of variables; $\fnt{eval_expr}$ evaluates a Clight expression to a value. These are defined as in CompCert. $\fnt{dom}$ returns the domain of a map as a list.  }
\label{fig:semantics-pr}
\end{center}
\end{DIFnomarkup}
\end{figure}

%% file: example_execution.tex
\label{sec:eg-exec}
We demonstrate how a program executes in our semantics in \Cref{fig:full-exec}. We show the program on the left, and each line's corresponding active threads on the right. The program begins with an initial thread $t_1$ in the thread pool and an initial team tree \ttree. When $t_1$ executes the parallel pragma on line 3 with the rule \textsc{Step-Parallel}, it forks another thread $t_2$ and makes a new team $\{t_1,t_2\}$. At this point, the variable name \code{r} in both threads refers to the same variable, and they can read \code{r} concurrently. Then  $t_1$ and $t_2$ each enter the \code{for} construct individually with the rule \textsc{Step-For}: the associated \textsc{Step-Priv} privatizes the reduction variable $\code{r}$ and initializes it to 0. Each thread is assigned part of the loop iterations (the assignment is set in the {\ForCtx} when the first thread reaches the \code{for} pragma), and running these iterations updates \code{r} to be $n$ in $t_1$ and $m$ in $t_2$ by the end of the loop. The first thread reaching line 10 is then paused at the initial barrier of \kw{SBRB}, until the other thread reaches line 10 and triggers \textsc{Step-Barrier}, and $t_2$ gives its permission on the original variable \code{r} to $t_1$. Then each thread executes \kw{SRed}, which in $t_2$ does nothing and in $t_1$ combines the threads' contributions $n$, $m$ with the original value $0$, setting the original version of \code{r} to $0+n+m$. After another barrier at the end of \kw{SBRB}, the \ForCtx's lifetime ends and the threads each execute \kw{SPrivEnd}, deallocating their private copies of \code{r} and returning the name \code{r} to refer to the original version. Therefore, after they exit the \code{for} construct, both threads will execute line 11 and print the same value $n+m$. 
Eventually they reach the end of the \!\code{parallel} construct on line 13, where the \kw{SBRB} ends the team, removes the team's contexts and nodes from the team tree, and transfers all of $t_2$'s permissions to \code{r} to $t_1$ (so now $t_1$ has full permissions). Finally, the program ends with just $t_1$ in the thread pool and a node for $t_1$ in \ttree, and we conclude that the program terminates successfully.

\begin{center}

\begin{tabular}{l@{}l}
\begin{minipage}[b]{0.6\textwidth}
{
\captionsetup[lstlisting]{margin = 5 mm}
\begin{lstlisting}[style=customc, linewidth = 6cm]
int r = 0;

#pragma omp parallel num_threads(2)
{
    // r is shared variable
#pragma omp for reduction(+:r)
    for(...) {
        // r is privatized
        r+=...
    }
    printf("%d\n",r); //prints n+m twice
    ...
}

// program ends
\end{lstlisting}
}
\captionof{figure}{An OpenMP program. \vspace{7em} }
\end{minipage}
&
\begin{minipage}[b]{0.45\textwidth}
\includegraphics[width=\linewidth]{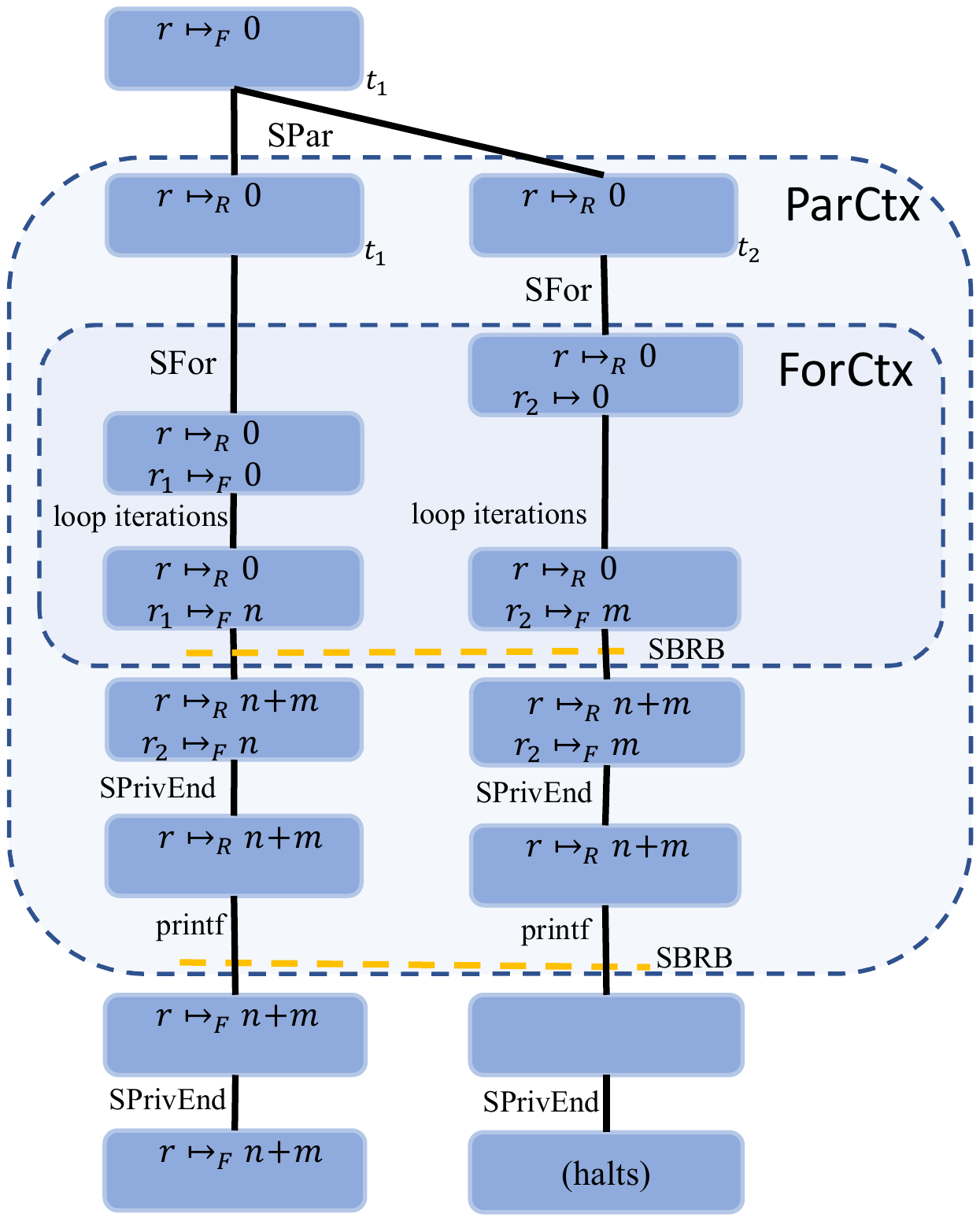}
\captionof{figure}{The execution flow and liftime of the team contexts. Each cell contains a thread's local view to the memory, as well as the permissions annotated as subscripts: $R$ for {\Readable}, and F for {\Freeable}.}
\end{minipage}

\end{tabular}
\label{fig:full-exec}

\end{center}

%% file: drf.tex
\label{sec:drf-proof}
We prove that any ClightOMP execution is race-free. We begin by defining the memory and synchronization events produced by a \lang execution.

\begin{figure}
  \small \centering
\begin{align*}
\text{(CompCert) } \MemEv \triangleq\ & \Alloc_i\ l\ |\ \Free_i\ l\ |\ \Read_i\ l\ |\ \Write_i\ l \\
\SyncEv \triangleq\ & \Par_i\ \vec{t}\ |\ \BarEv\ \vec{t} \\
\Event \triangleq\ & \MemEv\ |\ \SyncEv
\end{align*}

\[
\begin{array}{ll}
\StepClight \text{ in thread } i & \text{list of }\MemEv_i \\
\text{privatization/reduction steps in thread } i & \text{list of }{\MemEv_i}\> \\
\StepPar \text{ in thread } i \text{ spawning } \vec{t} & [\Par_i\ \vec{t}] \\
\StepBar \text{ synchronizing } \vec{t} & [\BarEv\ \vec{t}] \\
\text{other steps} &  [] \\
\end{array} 
\]
\caption{ClightOMP memory events emitted by a step.}
\label{fig:mem-ev}
\end{figure}

A ClightOMP step that accesses memory emits a list of CompCert memory events; a ClightOMP step that performs synchronization emits a synchronization event. \Cref{fig:mem-ev} shows the events emitted by each step. Per-thread Clight steps and \StepPriv /\StepPrivEnd /\StepRedL /\StepRedNL\ emit a sequence of memory events corresponding to their memory operations, each indexed by the thread number $i$ that runs the step and the accessed memory location $l$. Specifically, {\StepPriv} emits an {\Alloc} event for each memory allocation and {\StepPrivEnd} emits corresponding {\Free} events; {\StepRedL} emits {\Read} events for reading the private copies of the reduction variable, followed by a {\Write} event for updating the original copy. {\StepPar} emits an parallel event $\Par_i\ \vec{t}$ that synchronizes thread $i$ with the spawned threads $\vec{t}$.  {\StepBar} emits a barrier event that synchronizes among the team $\vec{t}$.
We define a labeled version of our semantics $\langle i \cdot \mho, \tp,  \ttree, m \rangle \xrightarrow{\vec{e}} \langle \mho, \tp',  \ttree', m' \rangle$ where each step is labeled with the events it emits; then the event trace of an execution is the concatenation of all events emitted at each step.




As is standard, we define a data race as a pair of conflicting memory events (i.e., events to the same location that are not both {\Read}s) that are not ordered via a happens-before relation~\cite{lamport-vector-clocks}. First, we need to define the happens-before relation:
\begin{defi}[happens-before] Consider a list $\vec{e}$ of events and two events $e_i, e_j \IN \vec{e}$ where $i < j$ are indexes to $\vec{e}$. Then $e_i$ happens before $e_j$, written $e_i <_{\mathrm{hb}} e_j$, if: 
\begin{itemize}
\item $e_i$ and $e_j$ are from the same thread, or
\item $e_i$ is $\Par\ \vec{t}$ and the thread of $e_j$ is in $\vec{t}$, or
\item one of $e_i$ and $e_j$ is $\BarEv\ \vec{t}$ and the thread of the other is in $\vec{t}$, or
\item there is another event $e_k$ such that $e_i <_\mathrm{hb} e_k$ and $e_k <_\mathrm{hb} e_j$.
\end{itemize}
\end{defi}

We observe that if $e_i <_\mathrm{hb} e_j$, then $e_i$ also happens before any following events by $e_j$'s thread; in fact, when two threads synchronize, all previous events by the first thread happen before all following events of the second thread. This makes it useful to talk about an event synchronizing with a thread, rather than simply with another event.

\begin{defi}[synchronizing with a thread] An event $e_i$ synchronizes with a thread $t$ in a trace $\vec{e}$ if for any event $e_t$ by $t$, if we extend $\vec{e}$ with a new trace $\vec{e'}$ including $e_t$, then $e_i <_\mathrm{hb} e_t$ in $\vec{e} +\mspace{-4mu}+\> \vec{e'}$.
\end{defi}

We can now begin to prove race-freedom. Intuitively, we can trace a permission throughout an execution: {\Writable} permission flows forward to \emph{all} threads that hold at least {\Readable} permission for the rest of the execution, and {\Readable} permission flows forward to \emph{some} thread that holds at least {\Readable} permission for the rest of the execution. We now formalize and prove this intuition in terms of event synchronization. 

\begin{lemma}[Read synchronization] \label{lemma:read-sync} Suppose we have an execution ending in a state $\mathit{st}$ whose trace contains an event $e = \Read_i\ l$, and $l$ is not deallocated in $\mathit{st}$, i.e., there is some thread that has some permission to $l$ in $\mathit{st}$. Then there is some thread $j$ such that 1) $j$ has at least {\Readable} permission to $l$ in $\mathit{st}$, and 2) $e$ synchronizes with $j$ in the trace.
\end{lemma}
\begin{proof} By induction on the execution. In the base case, the execution consists of a single step $\mathit{st}_0 \xrightarrow{\vec{e}} \mathit{st}$ where $\Read_i\ l \IN \vec{e}$. A thread only performs operations it has permissions to, and steps that perform {\Read}s do not change permissions, so thread $i$ still has {\Readable} permission to $l$ in $\mathit{st}$ and $\Read_i\ l$ synchronizes with $i$.

In the inductive case, we have an execution ending in a state $\mathit{st}$ with a thread $j$ as described, and we consider a next step $\mathit{st} \xrightarrow{e'} \mathit{st}'$. If $j$ still has {\Readable} permission to $l$ in $\mathit{st}'$, then we are finished. If it does not, then $e'$ must be emitted from a synchronization operation that changed $j$'s permissions. Whether this operation was a {\Par} or {\BarEv}, it synchronized with a set of threads $\vec{t}$ whose new permissions in $\mathit{st}'$ sum to the same result as the sum of the source permissions in $\mathit{st}$, which included $j$'s {\Readable} permission to $l$. Thus, there must be at least one thread $k \IN \vec{t}$ that has at least {\Readable} permission in $\mathit{st}'$, and since $e$ synchronized with $j$ and $e'$ synchronized $j$ with $k$, $e$ now synchronizes with $k$ in the extended trace.
\end{proof}

\begin{lemma}[Write synchronization] \label{lemma:write-sync} Suppose we have an execution ending in a state $\mathit{st}$ whose trace contains an event $e = \Write_i\ l$ (or {\Alloc} or {\Free}). Then for all threads $j$ that have at least {\Readable} permission to $l$ in $\mathit{st}$, $e$ synchronizes with $j$ in the trace.
\end{lemma}
\begin{proof}
By induction on the execution. In the base case, the execution consists of a single step $\mathit{st}_0 \xrightarrow{\vec{e}} \mathit{st}$ where $e \IN \vec{e}$. If $e$ is a {\Write} or {\Free} then thread $i$ must have had at least {\Writable} permissions to $l$ in $\mathit{st}_0$, so by coherence no other thread had {\Readable} permissions to $l$; if $e$ is {\Alloc} then its location must have not been allocated in $\mathit{st}_0$, so $i$ is the only thread with permissions to $l$ in $\mathit{st}$. In either case, $i$ is the only thread in $\mathit{st}$ with at {\Readable} permission to $l$, and $e$ synchronizes with $i$ since they are in the same thread.

In the inductive case, we have an execution ending in a state $\mathit{st}$ with threads synchronized with $e$ as described, and we consider a next step $\mathit{st} \xrightarrow{e'} \mathit{st}'$ and a thread $j$ with at least {\Readable} permission to $l$ in $\mathit{st}'$. If $j$ had {\Readable} permission to $l$ in $\mathit{st}$, then we are finished. If it did not, then $e'$ must be a synchronization operation that changed $j$'s permissions. Whether this operation was a {\Par} or {\BarEv}, it synchronized a set of threads $\vec{t}$ whose old permissions in $\mathit{st}$ summed to a result that included $j$'s new {\Readable} permission to $l$. Thus, there must be at least one thread $k \IN \vec{t}$ that had at least {\Readable} permission in $\mathit{st}$, and since $e$ synchronized with $k$ and $e'$ synchronized $k$ with $j$, $e$ now synchronizes with $j$ in the extended trace.
\end{proof}


These dual lemmas show that read or write/alloc/free events are always synchronized in some way with threads that hold {\Readable} permissions after them. We can use this to prove data race freedom:

\begin{theorem}[Data race freedom] If an execution's trace contains two conflicting events $e_1$ and $e_2$, then $e_1 <_\mathrm{hb} e_2$.
\end{theorem}
\begin{proof}
By induction, it suffices to consider the case where $e_2$ is produced by the last step of the execution, $\mathit{st} \rightarrow \mathit{st}'$, and the execution up through $\mathit{st}$ is race-free. Let $t_2$ be the thread of $e_2$. Since $e_1$ and $e_2$ conflict, they are on the same location $l$ and either $e_1$ or $e_2$ is a write/alloc/free. In the case where $e_1$ is a write, $t_2$ has at least {\Readable} permission to $l$ (since $e_2$ must be at least a read), and so by \cref{lemma:write-sync} $e_1$ synchronizes with $t_2$ and thus $e_1 <_\mathrm{hb} e_2$. In the case where $e_1$ is a read, $l$ is still allocated in $\mathit{st}$ (since we perform $e_2$ on it\footnote{Note that Clight never reallocates freed locations, so we do not need to worry about the case where $e_2$ {\Alloc}s $l$ after $e_1$ reads $l$.}), so by \Cref{lemma:read-sync} there must be some thread $j$ in $\mathit{st}$ with at least {\Readable} permission to $l$ such that $e_1$ synchronizes with $j$. But we know that $e_2$ must be either a {\Write} or a {\Free}, and in either case $t_2$ must have at least {\Writable} permission to $l$. By coherence, this means that no other thread has {\Readable} permission to $l$, so the thread $j$ must be exactly $t_2$; once again, $e_1$ synchronizes with $t_2$ and thus $e_1 <_\mathrm{hb} e_2$.
\end{proof}


Thus, any execution in our semantics is necessarily race-free.